\documentclass[aps,prl,reprint]{revtex4-2}

\usepackage{amsmath,amssymb,amsfonts,mathtools,bm}
\usepackage{amsthm,booktabs,array,graphicx,hyperref}

\newtheorem{theorem}{Theorem}
\newtheorem{proposition}{Proposition}
\newtheorem{lemma}{Lemma}
\newtheorem{corollary}{Corollary}

\theoremstyle{plain}

\begin{document}

\title{Quantum States over Spacetime Can Hide Setting-Dependent Dynamics}

\author{Jianqi Sheng}
\affiliation{Department of Physics, City University of Hong Kong, Hong Kong SAR, China}

\begin{abstract}
Determining whether correlations across quantum events arise from one physical process is essential for interpreting quantum states over spacetime (QSOSTs). Yet separately realizable settings can share the same unconditional QSOST without admitting a common process, because compression hides distinctions between underlying dynamics. Using least positive lifts, we derive an exact joint-domination criterion for conditioned data to admit a common deterministic parent. A reset--dephasing example and a sparse interferometric witness make this obstruction testable under trusted readout assumptions. We further establish a minimum-five-reference theorem: for two laboratories with qubit inputs and outputs, any four fixed local product-unitary reference groups cannot universally identify deterministic processes from exact joint data, whereas a specified set of five suffices. These results establish when reconstructed temporal correlations support interpretation as alternative readouts of one stable quantum process, making preparation--readout separation a testable physical property.
\end{abstract}

\maketitle

A quantum state has physical force because it separates preparation from readout.  Once $\rho$ is prepared, a later choice of POVM changes which property is read but not the state supplied to the measurement.  This separation licenses delayed measurement choice, comparison of different settings, and their use as data about one stable preparation.  The Hughston--Jozsa--Wootters theorem gives a related common-source statement: different ensembles of one density operator can be generated from one purification by alternative terminal measurements ~\cite{hughston1993complete}.  Analogous dilation statements hold for fully specified instrument and super-instrument assemblages ~\cite{piani2015channel,uola2018instruments,baroni2026composable}.

Quantum states over time and spacetime extend state-like descriptions to events that need not lie on one spatial slice ~\cite{fullwood2022qsot,lie2024unique,lie2025multipartite}.  Trusted causally agnostic interferometry can reconstruct every coefficient of a QSOST without presupposing an order between the laboratories ~\cite{lie2026qsost}.  Yet this operational completeness leaves a distinct physical question open: does the reconstructed object still separate a fixed process from a later choice of how its record is measured?

Consider an experiment with process slots $A,B$ and a terminal record.  The \emph{fixed-process hypothesis} says that the source and the dynamics on $A,B$ are established independently of a setting $x$, sampled after the slots close, and that $x$ controls only a terminal record measurement.  The word ``stable'' below means independent of this terminal setting, not static in time.  The competing explanation is not necessarily exotic: an earlier control, source--setting correlation, memory, drift, or leakage may select a different process context $W_x$.  Natural consistency checks are that every $x$ separately fits quantum theory and that averaging over the record gives the same reconstructed QSOST.  We ask whether those two checks are sufficient.

\begin{figure*}[t]
\centering
\includegraphics[width=0.8\textwidth]{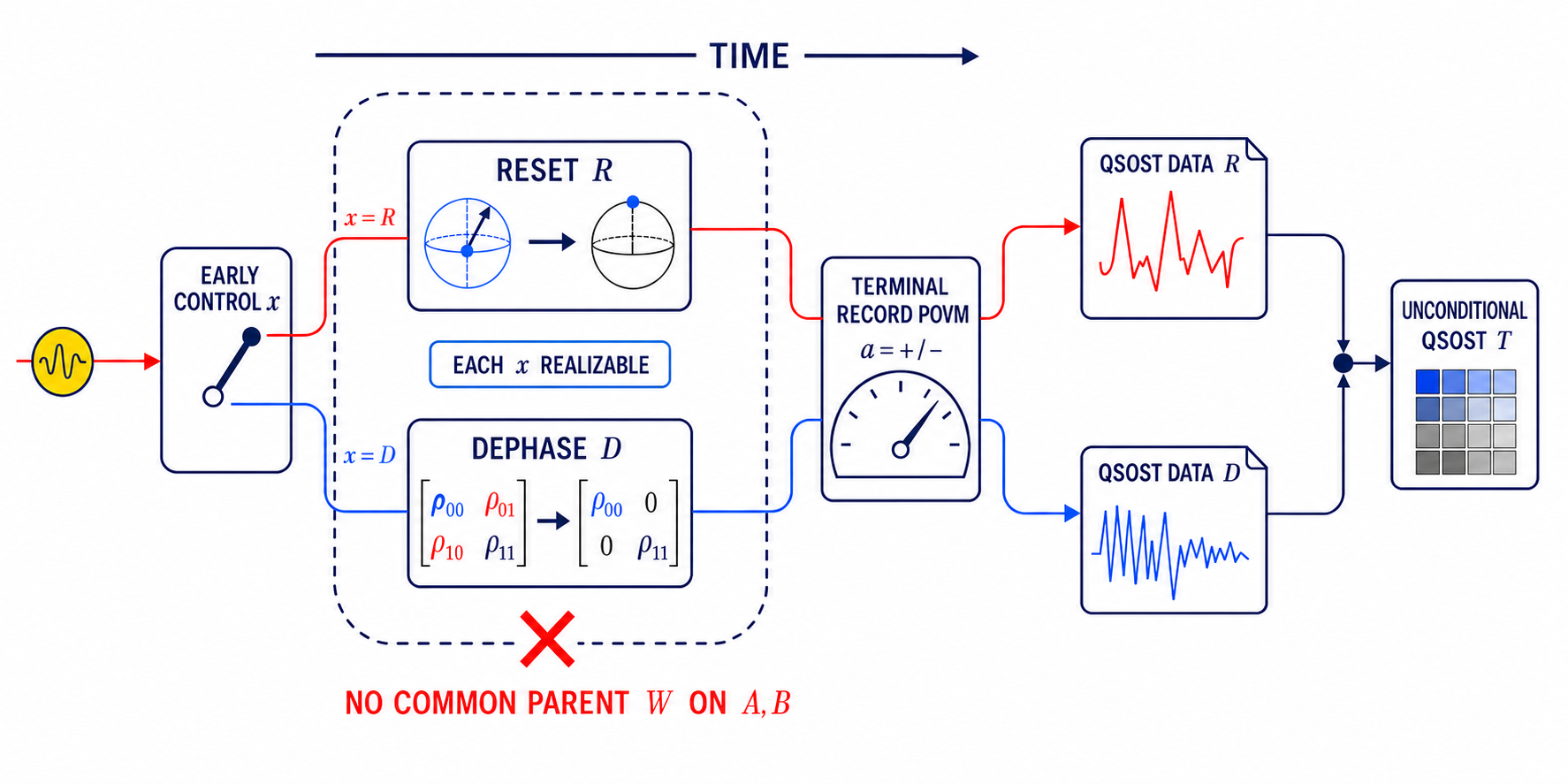}
\caption{\textbf{The same unconditional QSOST can hide setting-dependent dynamics.} On alternative runs, an early classical setting $x$ selects reset $R$ (upper red path) or complete dephasing $D$ (lower blue path).  Each setting is separately realizable by ordinary $A\prec B$ quantum dynamics, and both use the same terminal record POVM with outcomes $a=\pm$.  Joint probe--record data yield different conditioned QSOST families $\{\Sigma_{a|R}\}_a$ and $\{\Sigma_{a|D}\}_a$.  The navy arrows denote marginalization over $a$, not further physical evolution, and converge on the same unconditional QSOST, $\sum_a\Sigma_{a|R}=\sum_a\Sigma_{a|D}=T$; the operator grid depicts $T$ schematically, not as a Bloch state.  The crossed dashed enclosure denotes the obstruction: no $x$-independent deterministic parent $W$ on the declared $A,B$ slots generates both families solely through alternative terminal measurements.  A larger past-controlled quantum realization remains possible.}
\label{fig:setting_dependent_dynamics}
\end{figure*}

Figure~\ref{fig:setting_dependent_dynamics} shows why the two consistency checks are insufficient.  The interferometer reconstructs only the compressed branch operator $\Sigma_{a|x}=Q_+(G_{a|x})$.  Ignoring the terminal outcome can make the reconstructed sums identical even when their hidden positive process completions differ.  Conditioning on the records exposes those differences: each setting admits its own deterministic parent, but no choice of positive lifts for all settings fits one parent on the declared slots. The time arrow in the figure ends at the record POVM; the arrows to the common QSOST denote marginalization, not further evolution.  The family can therefore be implemented by early control selecting $W_x$, while failing the interpretation as alternative terminal measurements of one fixed process.

We develop three connected results for this common-process question: an exact criterion from the conditioned data, a reset--dephasing benchmark linking the counterexample to its visibility threshold and sparse witness, and the minimum reference count needed to distinguish deterministic qubit processes. Fully specified terminal branches obey the known terminal-party realization criterion~\cite{bavaresco2019sdi}; the issue here is whether a common parent survives QSOST compression.  Supplemental Material, Sec.~E, distinguishes this question from simultaneous purification and from channel, process, and spacetime marginal problems~\cite{hsieh2022marginal,jia2023spacetime}.

A process predicts responses to arbitrary admitted interventions on the declared slots; a QSOST records the complex response accessible to the trusted causally agnostic interferometer.

Consider two laboratories in the interleaved order $A_I A_O B_I B_O$.  Let $L_V$ be the standard orthogonal projector onto the bipartite valid-process subspace~\cite{oreshkov2012process, branciard2016witnesses}.  Deterministic processes form the compact set
\begin{equation}
\begin{aligned}
\mathfrak P &= \{W\succeq0:L_V(W)=W,\ \operatorname{Tr}W=d_O\}, \\
d_O&=d_{A_O}d_{B_O}.
\end{aligned}
\label{eq:process_set}
\end{equation}
We study the canonical QSOST compression, which requires
\begin{equation}
d_{A_I}=d_{A_O}, \qquad d_{B_I}=d_{B_O}.
\label{eq:equal_dimensions}
\end{equation}
Let $\Pi$ group $I=A_I B_I$ and $O=A_O B_O$, set $\overline G=\Pi G\Pi^\dagger$, and let $S_{I:O}$ swap the two isomorphic spaces.  In the standard-positive convention,
\begin{equation}
Q_+(G) = \operatorname{Tr}_{I} \left[ S_{I:O}\,\overline G^{T_I} \right].
\label{eq:Qplus}
\end{equation}
Equation~\eqref{eq:Qplus} is the precise sense in which a QSOST is a visible projection of a process.  Two positive processes in the same $Q_+$ fiber are indistinguishable through the declared trusted QSOST interface, even though the discarded directions may decide whether several conditioned records share one parent. The interferometric derivation and the convention conversion/process projector are given in Supplemental Material, Secs.~C and L, respectively.

A compressed family is \emph{settingwise realizable} if each $x$ has positive lifts with its own parent, and \emph{commonly realizable} if the same can be done with one parent:
\begin{align}
G_{a|x}&\succeq0,\quad \sum_aG_{a|x}=W_x\in\mathfrak P,
\label{eq:general_process_assemblage}\\
\widetilde G_{a|x}&\succeq0,\quad Q_+(\widetilde G_{a|x})=\Sigma_{a|x},\quad \sum_a\widetilde G_{a|x}=W\in\mathfrak P\quad\forall x.
\label{eq:common_process_assemblage}
\end{align}
In the first line $Q_+(G_{a|x})=\Sigma_{a|x}$ is understood.  The distinction is entirely in whether the hidden parent may depend on $x$.  Settingwise realizability is what a separate process-validation test can establish for each setting.  Common realizability is the stronger stability claim needed to regard $x$ as a terminal readout of one device, to pool the tested conditioned data into one fixed-parent model, or to treat the declared laboratory block as $x$-independent relative to those terminal settings.

These are not merely two mathematical feasibility classes.  For fully specified lifts, common realizability is equivalent to preparing one process extension before $x$ and allowing $x$ to choose only a common-future terminal POVM.  Settingwise realizability is equivalent to a larger quantum circuit in which an earlier classical register may select $W_x$.  Applying $Q_+$ to the branches preserves these existence statements.  A strict separation therefore distinguishes a \emph{readout setting} from a \emph{process-context control} at the level of the declared slots; the exact realization theorem is proved in Supplemental Material, Sec.~G.

Let $|\Omega\rangle=\Pi^\dagger\sum_j|j\rangle_I|j\rangle_O$ and define the interleaved vectorization
\begin{equation}
\widetilde{\mathfrak v}(\Sigma) = \Pi^\dagger \sum_{\alpha,\beta}\Sigma_{\alpha\beta}^{*} |\alpha\rangle_I|\beta\rangle_O.
\label{eq:vmap_interleaved}
\end{equation}
For every Hermitian $W$ in the valid process subspace,
\begin{equation}
\langle\Omega|W|\Omega\rangle = \frac{\operatorname{Tr}W}{d_O}.
\label{eq:normalization_functional}
\end{equation}

\begin{proposition}[Exact positive-lift reduction]
For every Hermitian $G$,
\begin{align}
\widetilde{\mathfrak v}(Q_+(G))&=G|\Omega\rangle,
\label{eq:column_identity}\\
\operatorname{Tr}Q_+(G)&=\langle\Omega|G|\Omega\rangle.
\label{eq:trace_Q}
\end{align}
A nonzero $\Sigma$ admits a positive lift if and only if $p=\operatorname{Tr}\Sigma$ is real and positive.  Writing $|q\rangle=\widetilde{\mathfrak v}(\Sigma)$, its unique least positive lift and every other lift are
\begin{align}
M(\Sigma)&=\frac{|q\rangle\langle q|}{p},
\label{eq:minimal_lift}\\
G&=M(\Sigma)+H,\qquad H\succeq0,\quad H|\Omega\rangle=0.
\label{eq:all_lifts}
\end{align}
For $\Sigma=0$, set $M(0)=0$.
\end{proposition}

The proposition separates visible and invisible process content.  The data fix the least lift $M(\Sigma)$, while $H$ is hidden positive branch freedom that is completely invisible to the QSOST interface.  Different record settings can therefore agree on every unconditional QSOST observable yet demand least lifts that cannot be contained in one normalized process.  Under
\begin{equation}
\sum_a\Sigma_{a|x}=T,\qquad \operatorname{Tr}T=1\quad\forall x,
\label{eq:process_common_qsost}
\end{equation}
first require every nonzero branch to satisfy $p_{a|x}=\operatorname{Tr}\Sigma_{a|x}\in\mathbb R_{>0}$; otherwise the family is not even settingwise realizable.  Omit zero branches and construct directly from the remaining data
\begin{align}
p_{a|x}&=\operatorname{Tr}\Sigma_{a|x},\qquad |q_{a|x}\rangle=\widetilde{\mathfrak v}(\Sigma_{a|x}), \nonumber\\
R_x&=\sum_a\frac{|q_{a|x}\rangle\langle q_{a|x}|}{p_{a|x}}, \nonumber\\
Q_+(R_x)&=T,\qquad \langle\Omega|R_x|\Omega\rangle=1.
\label{eq:Rx}
\end{align}

\begin{proposition}[Process-domination criterion]
Settingwise and common realizability are respectively equivalent to
\begin{align}
&\exists\,W_x\in\mathfrak P:\quad W_x\succeq R_x \quad\text{for every }x,
\label{eq:settingwise_domination}\\
&\exists\,W\in\mathfrak P:\quad W\succeq R_x \quad\text{for every }x.
\label{eq:common_domination}
\end{align}
\end{proposition}

Equivalently, define
\begin{equation}
\begin{aligned}
\mu(\{R_x\}) &=\min_W\operatorname{Tr}W, \\[-2pt]
&\text{subject to}\quad L_V(W)=W, \\[-2pt]
&\hspace{5.9em}W\succeq R_x\quad\forall x.
\end{aligned}
\label{eq:mu_primal}
\end{equation}
\begin{equation}
\text{common realizability} \quad\Longleftrightarrow\quad \mu(\{R_x\})=d_O.
\label{eq:mu_criterion}
\end{equation}
Indeed, every feasible $W$ is positive and obeys $\operatorname{Tr}W/d_O=\langle\Omega|W|\Omega\rangle \ge\langle\Omega|R_x|\Omega\rangle=1$; the minimum is attained on a compact positive trace sublevel.  Hence equality holds exactly when the minimizer is a deterministic process. It is useful to report the dimensionless fixed-process gap
\begin{equation}
g_{\mathrm{fix}} = \frac{\mu(\{R_x\})}{d_O}-1 \ge0.
\label{eq:fixed_process_gap}
\end{equation}
This is a compatibility gap, not a general process resource monotone. Conditioned QSOST tomography gives $g_{\mathrm{fix}}=0$ exactly when the tested data admit at least one fixed parent.  This establishes compatibility, not a unique process or proof that $x$ was actually chosen late.  A positive gap rejects every $x$-independent parent on the declared slots even when each setting separately remains quantum-valid.  A joint physical account then requires an enlarged model with explicit setting-dependent control, or failure of at least one interface, independence, or isolation assumption. The proofs, dual program, and generic witnesses are in Supplemental Material, Secs.~A and F.

Call the compression \emph{common-parent faithful} when every finite family with one unconditional QSOST and settingwise process realizations admits one deterministic parent on the declared slots.  The next result identifies exactly when this property holds.

\begin{lemma}[Full exposure]
\label{lem:full_exposure}
Every rank-$r$ deterministic process $W$ admits vectors $|v_a\rangle$ such that
\begin{align}
W&=\sum_{a=0}^{r-1}|v_a\rangle\langle v_a|, \qquad |\langle v_a|\Omega\rangle|^2=\frac1r,
\label{eq:rank_one_decomposition}\\
M\!\left(Q_+(|v_a\rangle\langle v_a|)\right) &=|v_a\rangle\langle v_a|.
\label{eq:branch_exposed}
\end{align}
\end{lemma}

Thus a Fourier-rotated rank-one decomposition exposes any deterministic process as the sum of its own least branch lifts.  The construction is given in Supplemental Material, Sec.~A.

\begin{theorem}[Fixed-process faithfulness criterion]
\label{thm:compression_gluing}
Let $\mathcal C_{\mathrm{settingwise}}$ be the finite families obeying Eq.~\eqref{eq:process_common_qsost} that satisfy Eq.~\eqref{eq:settingwise_domination}, and let $\mathcal C_{\mathrm{common}}$ satisfy Eq.~\eqref{eq:common_domination}. Then
\begin{equation}
Q_+|_{\mathfrak P}\text{ is injective} \quad\Longleftrightarrow\quad \mathcal C_{\mathrm{settingwise}} = \mathcal C_{\mathrm{common}}.
\label{eq:compression_gluing_equivalence}
\end{equation}
\end{theorem}

\begin{proof}
For a settingwise realization, $Q_+(W_x)=\sum_a\Sigma_{a|x}=T$.  If $Q_+|_{\mathfrak P}$ is injective, all $W_x$ coincide, so the same branch lifts already give a common realization.

Conversely, let $W_0\ne W_1$ be deterministic processes with $Q_+(W_0)=Q_+(W_1)$.  Apply Lemma~\ref{lem:full_exposure} to each and set $\Sigma_{a|x}=Q_+(|v_{a|x}\rangle\langle v_{a|x}|)$, padding the smaller family with zero outcomes if necessary.  The resulting family is settingwise realizable and has one unconditional QSOST.  Were it commonly realizable by $W$, the least-lift reduction would give $W\succeq\sum_a|v_{a|x}\rangle\langle v_{a|x}|=W_x$ for both $x$. Since $W$ and $W_x$ have the same trace $d_O$, positivity forces $W=W_x$ for both settings, contradicting $W_0\ne W_1$.
\end{proof}

The criterion also has an exact finite multi-reference version.  For fixed factorwise references $\mathcal J=\{J_1,\ldots,J_K\}$, write $r_\ell=|\widetilde\Omega_{J_\ell}\rangle$ and $R=[r_1\ \cdots\ r_K]$.  The joint data are equivalent to the visible column array $GR$.  Every finite settingwise-realizable conditioned family with one unconditional joint compression has a common deterministic parent if and only if $W\mapsto WR$ is injective on $\mathfrak P$, equivalently if $\operatorname{ran}R$ separates the traceless Hermitian valid-process directions.

For two laboratories whose input and output slots are all qubits, this criterion has a sharp reference count: any four fixed local product-unitary references leave distinct deterministic processes indistinguishable, whereas the five Pauli pairs
\begin{equation*}
(I,I),\quad (Z,Y),\quad (Y,Z),\quad (Y,X),\quad (X,I)
\end{equation*}
suffice.  The lower bound covers arbitrary allowed unitaries, including complex non-Pauli and dependent references; the upper bound has an exact rank-$87$ certificate.  Thus five reference groups are necessary and sufficient for deterministic-process identification from exact joint data, and hence for universal common-parent faithfulness in this model.  A group supplies its full interferometric data; this is neither a claim of five measurement settings nor of arbitrary branch reconstruction.  The theorem, analytic lower bound, and reproducible certificate are in Supplemental Material, Sec.~O, together with the general positive-lift and domination criteria.

Common realizability supplies consistency for the tested finite family of terminal settings on the declared slots.  It does not certify the actual mechanism or guarantee stability under untested controls, ancillary extensions, or arbitrary network composition.

Kernels, ranges, and dimensions involving Hermitian operators are taken over the real vector space of Hermitian matrices.  Write
\begin{equation}
m=d_{A_I}=d_{A_O}, \qquad n=d_{B_I}=d_{B_O}.
\label{eq:mn_dimensions}
\end{equation}

\begin{proposition}[Exact hidden-fiber dimension]
\label{prop:general_kernel_dimension}
For every $m,n\ge1$,
\begin{equation}
\begin{aligned}
&\dim_{\mathbb R} \left[ \ker Q_+\cap\operatorname{ran}L_V\cap\operatorname{Herm} \right] \\
&\qquad = (m^2-1)(n^2-1)(m^2+n^2-1).
\end{aligned}
\label{eq:general_kernel_dimension}
\end{equation}
\end{proposition}

An operator-basis proof is given in Supplemental Material, Sec.~M.  In particular, the hidden fiber vanishes only if one laboratory is one-dimensional.

\begin{corollary}[Common-parent obstruction with locally matched dimensions]
\label{cor:universal_gluing_failure}
For every nontrivial equal-local-dimension bipartite scenario $m,n\ge2$,
\begin{equation}
\mathcal C_{\mathrm{common}} \subsetneq \mathcal C_{\mathrm{settingwise}}.
\label{eq:universal_strict_inclusion}
\end{equation}
If $m=1$ or $n=1$, the two classes coincide.
\end{corollary}

\begin{proof}
The dimension in Eq.~\eqref{eq:general_kernel_dimension} is positive exactly when $m,n\ge2$.  If $\Delta$ is a nonzero Hermitian vector in this kernel, Eqs.~\eqref{eq:normalization_functional} and \eqref{eq:trace_Q} imply $\operatorname{Tr}\Delta=0$.  Therefore the full-rank process $W_*=I/(mn)$ belongs to $\mathfrak P$, and $W_*\pm\epsilon\Delta$ are distinct deterministic processes in one $Q_+$ fiber for all sufficiently small $\epsilon>0$. Theorem~\ref{thm:compression_gluing} converts them into a strict conditioned separation.  If one laboratory is one-dimensional, the hidden fiber is zero, so the injective direction of that theorem gives equality.
\end{proof}

The corollary asserts the existence of incompatible conditioned families; it does not say that every QSOST family is incompatible.

The result is constructive: Supplemental Material, Sec.~M, gives distinct full-rank $B\prec A$ deterministic processes in the same QSOST fiber for every $m,n\ge2$.  For bipartite qubits the hidden Hermitian fiber already has dimension $63$.

\paragraph{A reset--dephasing benchmark.---}

Consider two ordinary qubit channels from $A_O$ to $B_I$ in the definite order $A\prec B$: reset erases its input and prepares $|0\rangle$, while complete dephasing retains only computational-basis information.  Two equiprobable terminal outcomes expose different hidden coherences of these processes, leaving the same unconditional QSOST $|00\rangle\langle00|$. The explicit branches, positive lifts, and deterministic $A\prec B$ parents are given in Supplemental Material, Sec.~A.  Each setting is therefore separately process realizable.

\begin{proposition}[Minimal two-outcome common-parent separation]
The two-setting, two-outcome assemblage described above has no common deterministic process parent.  Its exact process-domination cost is
\begin{equation}
\mu=8,
\label{eq:minimal_mu_eight}
\end{equation}
whereas a deterministic bipartite-qubit process has trace four.
\end{proposition}

An explicit dual certificate and matching primal operator are given in Supplemental Material, Sec.~A. Operationally, both settings pass separate process validation, whereas their joint description would require twice the allowed deterministic-process normalization, so $g_{\mathrm{fix}}=1$.  The failure is therefore specifically a failure of the fixed-parent model on the declared slots, not of quantum realizability.

Two settings are necessary for any gluing question.  With only one outcome per setting, the conditioned operator is just the common unconditional QSOST, and any one of its deterministic lifts is automatically a common parent.  The construction above is therefore minimal in the numbers of settings and outcomes.  We do not claim simultaneous minimality of Hilbert-space dimension, circuit size, or tomography cost.

To connect this same example to its measured signal, scale every record-dependent coherence by a visibility $0\le\eta\le1$.  The explicit dual and primal operators in Supplemental Material, Sec.~A, give
\begin{align}
\mu(\eta)&=\max\{4,8\eta^2\},
\label{eq:eta_exact_cost}\\
g_{\mathrm{fix}}(\eta)&=\max\{0,2\eta^2-1\},
\label{eq:eta_fixed_process_gap}\\
\text{common-process realizable} &\iff \eta\le\frac1{\sqrt2}.
\label{eq:eta_exact_threshold}
\end{align}
Below the threshold, one deterministic process can absorb both conditioned descriptions; above it, the record-dependent coherences are too strong to be embedded in any common parent.  The threshold therefore quantifies the experimental visibility required to distinguish genuine common-parent failure from noise that still admits a fixed-parent explanation.

A tangent certificate at the threshold yields a sparse linear witness. Let $I_{PQ,a|x}$ be the branch-resolved complex interference coefficient for local Pauli pair $P,Q$, and define
\begin{equation}
\Delta_x I_{PQ}=I_{PQ,+|x}-I_{PQ,-|x}.
\end{equation}
Direct expansion in Pauli operators gives
\begin{align}
\mathcal W =\frac14\bigl[{} &\operatorname{Re}\Delta_{\mathrm R}I_{XI} +\operatorname{Re}\Delta_{\mathrm R}I_{XZ} \nonumber\\
&+\operatorname{Im}\Delta_{\mathrm R}I_{YI} +\operatorname{Im}\Delta_{\mathrm R}I_{YZ} \nonumber\\
&+\operatorname{Re}\Delta_{\mathrm D}I_{XX} +\operatorname{Re}\Delta_{\mathrm D}I_{YY} \nonumber\\
&-\operatorname{Im}\Delta_{\mathrm D}I_{XY} +\operatorname{Im}\Delta_{\mathrm D}I_{YX} \bigr].
\label{eq:sparse_W_interference}
\end{align}
Every common-process assemblage satisfies
\begin{equation}
\boxed{\mathcal W\le\sqrt2.}
\label{eq:sparse_W_bound}
\end{equation}
For the visibility family, $\mathcal W=2\eta$, so this witness detects the entire non-common region $\eta>1/\sqrt2$.  Its tangent derivation is given in Supplemental Material, Sec.~A. Thus, for this binary qubit family and under the trusted readout model below, the strict separation can be witnessed from eight branch-resolved interferometric quadratures rather than complete process or QSOST tomography. For this visibility family the witness is exact.  For arbitrary conditioned data, failure to violate it does not establish a common parent; the complete test is the tomography-level criterion $\mu=d_O$.

The full-rank white-noise continuation and an explicit controlled-circuit realization are given in Supplemental Material, Secs.~H and I, respectively.

Figure~\ref{fig:controlled_benchmark} implements the same family as an elementary two-slot circuit.  An early bit $x$ selects the reset or dephasing Stinespring isometry; the same binary POVM reads their future environment. Probe qubits coherently compare reference and local Pauli-intervention arms and are measured in logical $X/Y$ bases.  The exact isometries and branch calculation are given in Supplemental Material, Sec.~I.

\begin{figure}[t]
\centering
\includegraphics[width=\columnwidth]{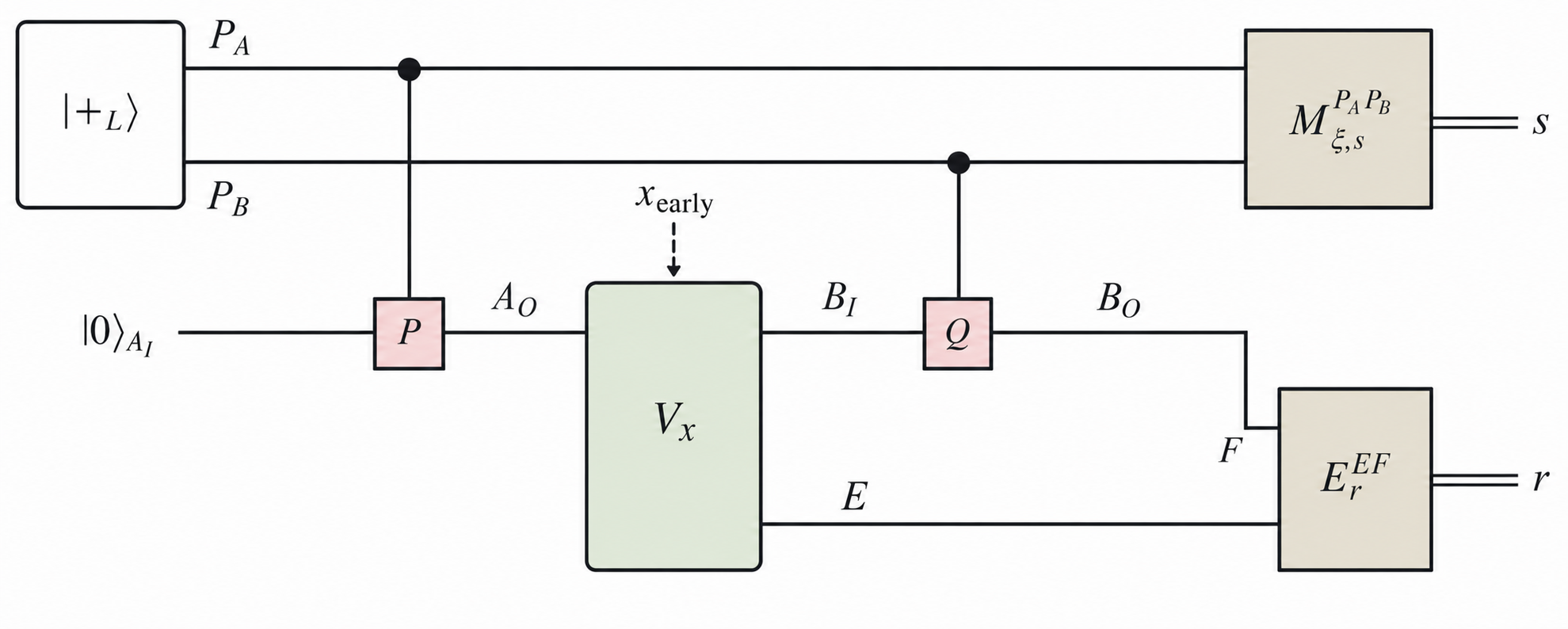}
\caption{\textbf{Past-controlled definite-order benchmark and factorized readout.} The logical probes $P_A,P_B$ start in $|+_L\rangle=(|00\rangle+|11\rangle)/\sqrt2$ and coherently control the local Pauli interventions $P,Q$.  An early classical bit $x\in\{\mathrm R,\mathrm D\}$ selects the Stinespring isometry $V_x:A_O\to B_I E$, realizing reset or complete dephasing between the slots. The terminal discard of $B_O$ is purified into the future system $F$. In the common future, the probe POVM $\{M_{\xi,s}^{P_AP_B}\}_{s=\pm1}$ uses the logical basis $\xi\in\{X,Y\}$, while the fixed record POVM $\{E_r^{EF}\}_{r=\pm1}$ is used for both values of $x$, with $r$ encoding the branch outcome $a$. The two measurement boxes form a product readout with no feed-forward. Their joint statistics $p(s,r|x,P,Q,\xi)$ supply the eight correlators in Eq.~\eqref{eq:sparse_W_interference}.  The circuit realizes the past-controlled alternative, not the fixed-process late-choice null model.}
\label{fig:controlled_benchmark}
\end{figure}

The ideal circuit attains $\mathcal W=2>\sqrt2$; reducing its record-dependent coherences gives the visibility family above.  The counterexample, its sharp $1/\sqrt2$ threshold, and the eight-quadrature witness thus describe the same physical benchmark.

\paragraph{Audit of a claimed late readout.---}
The witness tests a fixed-process future-choice model under explicit trusted conditions: the source is independent of the early interferometric label; fresh $x$ is generated after $A,B$ close, independently of the source and earlier device states, and selects only a terminal POVM on one fixed process extension.  The QSOST interface is calibrated and the probe--record readout is isolated and factorized, without leakage, feed-forward, postselection, or setting-dependent recalibration.  Such a model obeys $\mu(\{R_x\})=d_O$ and $\mathcal W\le\sqrt2$.  A calibrated finite-sample violation excludes that model on the declared slots, but does not identify the mechanism, certify causal order or device independence, or determine when a fact forms.  Nonrejection establishes compatibility only.  With finite tomography, the confidence region---not merely a plug-in estimate--- must be disjoint from the common-realizable set.  The realization proof and the operational validation and diagnosis are in Supplemental Material, Secs.~G and J.

The conditioned data admit a common process exactly when their least positive branch lifts are jointly dominated by one normalized valid process. For the reset--dephasing benchmark, the resulting cost $\mu(\eta)=\max\{4,8\eta^2\}$ gives the sharp $70.7\%$ visibility threshold; the same boundary is detected by its eight-quadrature witness.

For two laboratories whose inputs and outputs are all qubits, the ambiguity can be removed with a minimum of five fixed local product-unitary reference groups.  Any four allowed groups fail to guarantee deterministic-process identification, whereas the specified pairs $(I,I)$, $(Z,Y)$, $(Y,Z)$, $(Y,X)$ and $(X,I)$ suffice, with an exact rank-$87$ certificate.  This reference count assumes exact joint data and concerns deterministic processes.  Each group supplies its full interferometric data.  The five groups are not five measurement settings, and the theorem gives neither complete reconstruction of arbitrary branches nor a reference-count guarantee for noisy data.

\onecolumngrid
\clearpage
\begin{center}
\textbf{\large Supplemental Material for}\\[3pt]
\textit{Quantum States over Spacetime Can Hide Setting-Dependent Dynamics}
\end{center}

\appendix
\setcounter{secnumdepth}{1}
\setcounter{theorem}{1}
\setcounter{proposition}{4}
\setcounter{lemma}{1}
\setcounter{corollary}{1}

\section{Proofs of the main structural statements}
\label{app:main_proofs}

\begin{proof}[Proof of the column identity]
In grouped order set $|\Omega_{\mathrm g}\rangle=\sum_j|j\rangle_I|j\rangle_O$ and $\mathfrak v_{\mathrm g}(\Sigma)= \sum_{\alpha,\beta}\Sigma_{\alpha\beta}^{*} |\alpha\rangle_I|\beta\rangle_O$. In the grouped basis, write $\overline G_{i\alpha,j\beta}=\langle i,\alpha|\overline G|j,\beta\rangle$. Equation~(3) of the main text gives
\begin{equation}
[Q_+(G)]_{\alpha\beta} = \sum_i\overline G_{ii,\alpha\beta}.
\end{equation}
Hermiticity then implies
\begin{equation}
[\mathfrak v_{\mathrm g}(Q_+(G))]_{\alpha\beta} = \sum_i\overline G_{\alpha\beta,ii} = [\overline G|\Omega_{\mathrm g}\rangle]_{\alpha\beta}.
\end{equation}
Applying $\Pi^\dagger$ proves Eq.~(8) of the main text; taking the overlap with $\langle\Omega|$ proves Eq.~(9) of the main text.
\end{proof}

\begin{proof}[Proof of the exact positive-lift reduction]
If $G\succeq0$, then $p=\langle\Omega|G|\Omega\rangle\ge0$.  If $p=0$, positivity implies $G|\Omega\rangle=0$, so the column identity gives $\Sigma=0$. Conversely, $G=0$ lifts $\Sigma=0$, and every positive lift of the zero branch annihilates $|\Omega\rangle$.

Suppose $p>0$.  Since $\langle q|\Omega\rangle=p$, $M(\Sigma)|\Omega\rangle=|q\rangle$ and hence $Q_+(M(\Sigma))=\Sigma$.  For any other positive lift $G$, let
\begin{equation}
|u\rangle=\frac{G^{1/2}|\Omega\rangle}{\sqrt p}.
\end{equation}
Then $\langle u|u\rangle=1$ and
\begin{align}
G-M(\Sigma) &= G^{1/2}(I-|u\rangle\langle u|)G^{1/2} \succeq0.
\label{eq:minimal_lift_proof}
\end{align}
The difference has zero expectation on $|\Omega\rangle$ and is positive, so it annihilates $|\Omega\rangle$.  The converse follows from Eqs.~(8) and (11) of the main text.
\end{proof}

\begin{proof}[Proof of the process-domination criterion]
If positive branch lifts $G_{a|x}$ sum to a deterministic process $W_x$, the least-lift lemma gives $G_{a|x}\succeq M(\Sigma_{a|x})$ and hence $W_x\succeq R_x$.

Conversely, let $W_x\succeq R_x$ and set $H_x=W_x-R_x\succeq0$.  Equations~(7) and (13) of the main text imply $\langle\Omega|H_x|\Omega\rangle=0$.  Positivity gives $H_x|\Omega\rangle=0$ and $Q_+(H_x)=0$.  For arbitrary $\lambda_{a|x}\ge0$ with $\sum_a\lambda_{a|x}=1$, set
\begin{equation}
G_{a|x} = M(\Sigma_{a|x})+\lambda_{a|x}H_x.
\end{equation}
These operators are positive, have the prescribed QSOST images, and sum to $W_x$.  The same argument with one $W$ proves the common-parent statement.
\end{proof}

\begin{proof}[Proof of the full-exposure lemma]
The vector $|r_W\rangle=W^{1/2}|\Omega\rangle$ lies in $\operatorname{supp}W$ and has unit norm.  Choose an orthonormal basis $\{|e_k\rangle\}_{k=0}^{r-1}$ of $\operatorname{supp}W$ with $|e_0\rangle=|r_W\rangle$, and define
\begin{equation}
|u_a\rangle = \frac1{\sqrt r} \sum_{k=0}^{r-1} \exp\left(\frac{2\pi iak}{r}\right)|e_k\rangle.
\end{equation}
Set $|v_a\rangle=W^{1/2}|u_a\rangle$.  Completeness gives Eq.~(19) of the main text, while $\langle v_a|\Omega\rangle=\langle u_a|r_W\rangle=1/\sqrt r$. For $G_a=|v_a\rangle\langle v_a|$, one has $|q_a\rangle=G_a|\Omega\rangle$ and $p_a=|\langle v_a|\Omega\rangle|^2$, so Eq.~(10) of the main text gives Eq.~(20) of the main text.
\end{proof}

\subsection{Minimal definite-order construction and threshold}

In the interleaved order $A_IA_OB_IB_O$, define
\begin{equation}
|q\rangle=|0000\rangle, \qquad |u_{\mathrm R}\rangle=|0100\rangle, \qquad |u_{\mathrm D}\rangle=|0110\rangle.
\label{eq:q_u_minimal}
\end{equation}
The two equiprobable QSOST branches are
\begin{align}
\Sigma_{\pm|\mathrm R} &= \frac12\left( |00\rangle\langle00| \pm|00\rangle\langle10| \right),
\label{eq:Sigma_R_minimal}\\
\Sigma_{\pm|\mathrm D} &= \frac12\left( |00\rangle\langle00| \pm|01\rangle\langle10| \right).
\label{eq:Sigma_D_minimal}
\end{align}
Their least-lift sums are
\begin{align}
R_{\mathrm R} &=|q\rangle\langle q|+|u_{\mathrm R}\rangle\langle u_{\mathrm R}|,
\label{eq:RR_minimal}\\
R_{\mathrm D} &=|q\rangle\langle q|+|u_{\mathrm D}\rangle\langle u_{\mathrm D}|.
\label{eq:RD_minimal}
\end{align}
The visibility family and its domination operators are
\begin{align}
\Sigma_{\pm|\mathrm R}^{(\eta)} &= \frac12\left( |00\rangle\langle00| \pm\eta|00\rangle\langle10| \right),
\label{eq:eta_R}\\
\Sigma_{\pm|\mathrm D}^{(\eta)} &= \frac12\left( |00\rangle\langle00| \pm\eta|01\rangle\langle10| \right),
\label{eq:eta_D}\\
R_x^{(\eta)} &= |q\rangle\langle q|+\eta^2|u_x\rangle\langle u_x|.
\label{eq:R_eta}
\end{align}
The sparse vector functional used in the tangent derivation is
\begin{align}
\mathcal W={}& \operatorname{Re}\langle u_{\mathrm R}| (q_{+|\mathrm R}-q_{-|\mathrm R})\rangle \nonumber\\
&+ \operatorname{Re}\langle u_{\mathrm D}| (q_{+|\mathrm D}-q_{-|\mathrm D})\rangle.
\label{eq:sparse_W_vector}
\end{align}

The Choi operators of the reset and computational-basis dephasing channels are
\begin{align}
C_{\mathrm R}^{A_OB_I} &= I^{A_O}\otimes|0\rangle\langle0|^{B_I},
\label{eq:reset_choi}\\
C_{\mathrm D}^{A_OB_I} &= |00\rangle\langle00|+|11\rangle\langle11|.
\label{eq:dephasing_choi}
\end{align}
Define
\begin{align}
W_{\mathrm R} &= |0\rangle\langle0|^{A_I} \otimes C_{\mathrm R}^{A_OB_I} \otimes I^{B_O},
\label{eq:WR}\\
W_{\mathrm D} &= |0\rangle\langle0|^{A_I} \otimes C_{\mathrm D}^{A_OB_I} \otimes I^{B_O}.
\label{eq:WD}
\end{align}
These are deterministic $A\prec B$ process matrices of rank and trace four. Direct evaluation gives
\begin{equation}
Q_+(W_{\mathrm R})=Q_+(W_{\mathrm D})=|00\rangle\langle00|,
\label{eq:RD_same_Q}
\end{equation}
and the operators in Eqs.~\eqref{eq:RR_minimal} and \eqref{eq:RD_minimal} obey $R_x\preceq W_x$.

The branch vectors corresponding to Eqs.~\eqref{eq:Sigma_R_minimal} and \eqref{eq:Sigma_D_minimal} are
\begin{equation}
|q_{\pm|x}\rangle = \frac{|q\rangle\pm|u_x\rangle}{2}, \qquad \Sigma_{\pm|x} = \widetilde{\mathfrak v}^{-1}(|q_{\pm|x}\rangle).
\label{eq:two_outcome_branches}
\end{equation}
Their least lifts and explicit positive realizations are
\begin{align}
M_{\pm|x} &= 2|q_{\pm|x}\rangle\langle q_{\pm|x}|, \qquad M_{+|x}+M_{-|x}=R_x,
\label{eq:minimal_two_lifts}\\
G_{\pm|x} &= M_{\pm|x}+\frac12(W_x-R_x).
\label{eq:realizing_two_branches}
\end{align}
Because $(W_x-R_x)|\Omega\rangle=0$, these branches have the prescribed images and sum to $W_x$.

For the visibility family, the dual operators below give $\mu(\eta)\ge\max\{4,8\eta^2\}$.  If $\eta^2\le1/2$, the deterministic process
\begin{equation}
W(\eta)=(1-\eta^2)W_{\mathrm R}+\eta^2W_{\mathrm D}
\label{eq:W_eta_low}
\end{equation}
dominates both $R_x^{(\eta)}$ and has trace four.  If $\eta^2\ge1/2$, the valid operator
\begin{equation}
\eta^2(W_{\mathrm R}+W_{\mathrm D})
\label{eq:W_eta_high}
\end{equation}
dominates both and has trace $8\eta^2$.  This proves Eqs.~(26) and (28) of the main text. To make both domination claims explicit, set $t=\eta^2$ and $A_x=W_x-R_x\succeq0$.  For $t\le1/2$,
\begin{align}
W(\eta)-R_{\mathrm R}^{(\eta)} ={}& (1-2t)|u_{\mathrm R}\rangle\langle u_{\mathrm R}| +t|u_{\mathrm D}\rangle\langle u_{\mathrm D}| \nonumber\\
&+(1-t)A_{\mathrm R}+tA_{\mathrm D}\succeq0, \\
W(\eta)-R_{\mathrm D}^{(\eta)} ={}& (1-t)|u_{\mathrm R}\rangle\langle u_{\mathrm R}| \nonumber\\
&+(1-t)A_{\mathrm R}+tA_{\mathrm D}\succeq0.
\end{align}
For $t\ge1/2$,
\begin{align}
t(W_{\mathrm R}+W_{\mathrm D})-R_{\mathrm R}^{(\eta)} ={}& (2t-1)|q\rangle\langle q| +t|u_{\mathrm D}\rangle\langle u_{\mathrm D}| \nonumber\\
&+t(A_{\mathrm R}+A_{\mathrm D})\succeq0, \\
t(W_{\mathrm R}+W_{\mathrm D})-R_{\mathrm D}^{(\eta)} ={}& (2t-1)|q\rangle\langle q| +t|u_{\mathrm R}\rangle\langle u_{\mathrm R}| \nonumber\\
&+t(A_{\mathrm R}+A_{\mathrm D})\succeq0.
\end{align}

For completeness, at $\eta_0=1/\sqrt2$ the tangent reference is
\begin{equation}
|q^0_{a|x}\rangle = \frac{|q\rangle+a\eta_0|u_x\rangle}{2}, \qquad p^0_{a|x}=\frac12.
\end{equation}
The dual operators obey $\langle q^0_{a|x}|Y_x|q^0_{a|x}\rangle=\eta_0^2$. Summing Eq.~\eqref{eq:linear_witness} over both outcomes and settings gives
\begin{equation}
8\eta_0\mathcal W-8\eta_0^2\le d_O=4,
\label{eq:sparse_W_tangent_derivation}
\end{equation}
which is Eq.~(31) of the main text.

\begin{proof}[Proof of the minimal two-outcome separation]
Define the positive dual operators
\begin{align}
Y_{\mathrm R} &=4|0100\rangle\langle0100|,
\label{eq:YR_minimal}\\
Y_{\mathrm D} &=4\bigl( |0110\rangle\langle0110| +|1100\rangle\langle1100| +|1110\rangle\langle1110| \bigr).
\label{eq:YD_minimal}
\end{align}
Their sum is
\begin{equation}
Y_{\mathrm R}+Y_{\mathrm D} = 4I^{A_I}\otimes|1\rangle\langle1|^{A_O} \otimes I^{B_I}\otimes|0\rangle\langle0|^{B_O},
\end{equation}
and direct substitution in Eq.~\eqref{eq:LV} gives $L_V(Y_{\mathrm R}+Y_{\mathrm D})=I$.  Hence the pair is feasible for the dual SDP in Eqs.~\eqref{eq:mu_dual_objective} and \eqref{eq:mu_dual_constraints}, with value
\begin{equation}
\operatorname{Tr}(Y_{\mathrm R}R_{\mathrm R}) + \operatorname{Tr}(Y_{\mathrm D}R_{\mathrm D}) =8.
\end{equation}
Thus $\mu\ge8$.  Conversely, $W_{\mathrm R}+W_{\mathrm D}$ belongs to the valid process linear subspace, dominates both $R_{\mathrm R}$ and $R_{\mathrm D}$, and has trace eight.  Therefore $\mu=8$.
\end{proof}

\section{One-step record-conditioned QSOST assemblages}
\label{app:one_step}

\subsection{Choi convention and conditioned branches}

For a linear map $\mathcal E_{B|A}:\mathcal L(\mathcal H_A)\to\mathcal L(\mathcal H_B)$, we use the Choi convention~[S1]
\begin{equation}
C_{B|A}[\mathcal E] = \sum_{i,j}|i\rangle\langle j|_A\otimes \mathcal E_{B|A}(|i\rangle\langle j|_A).
\label{eq:choi}
\end{equation}
Complete positivity and trace preservation are equivalent to
\begin{equation}
C_{B|A}[\mathcal E]\succeq0, \qquad \operatorname{Tr}_B C_{B|A}[\mathcal E]=I_A.
\label{eq:choi_cptp}
\end{equation}
The map action is
\begin{equation}
\mathcal E_{B|A}(X) = \operatorname{Tr}_A[(X^T\otimes I_B)C_{B|A}[\mathcal E]].
\label{eq:choi_action}
\end{equation}

For each record setting $x$, let $\{\mathcal E_{a|x}\}_a$ be completely positive maps satisfying
\begin{equation}
\sum_a\mathcal E_{a|x}=\mathcal E \quad\text{for every }x,
\label{eq:common_channel}
\end{equation}
where $\mathcal E$ is trace preserving and independent of $x$.  With $C_{a|x}=C[\mathcal E_{a|x}]$, define
\begin{equation}
\Sigma_{a|x} = (\rho_A\otimes I_B)C_{a|x}^{T_A}.
\label{eq:one_step_branch}
\end{equation}
The traces are the branch probabilities,
\begin{equation}
p(a|x)=\operatorname{Tr}\Sigma_{a|x},
\label{eq:branch_probability}
\end{equation}
and
\begin{equation}
\sum_a\Sigma_{a|x}=T_{AB} \quad\text{independent of }x.
\label{eq:one_step_common_sum}
\end{equation}
The non-normalized form is essential because Eq.~\eqref{eq:one_step_branch} remains linear.

\subsection{Interferometric reconstruction}

Choose orthogonal unitary operator bases $\{U_\mu^A\}_{\mu=1}^{d_A^2}$ and $\{W_\nu^B\}_{\nu=1}^{d_B^2}$ with
\begin{equation}
\operatorname{Tr}[(U_\mu^A)^\dagger U_{\mu'}^A]=d_A\delta_{\mu\mu'}, \qquad \operatorname{Tr}[(W_\nu^B)^\dagger W_{\nu'}^B]=d_B\delta_{\nu\nu'}.
\end{equation}
The conditioned interference amplitudes are
\begin{equation}
I_{\mu\nu,a|x} = \operatorname{Tr}[(U_\mu^A\otimes W_\nu^B)\Sigma_{a|x}],
\label{eq:interference_coefficients}
\end{equation}
and
\begin{equation}
\Sigma_{a|x} = \frac{1}{d_Ad_B} \sum_{\mu,\nu} I_{\mu\nu,a|x}(U_\mu^A)^\dagger\otimes(W_\nu^B)^\dagger.
\label{eq:tomography}
\end{equation}
Both quadratures are generally required.  If a probe qubit coherently controls the reference and intervention arms, joint probe-record frequencies give
\begin{align}
\operatorname{Re}I_{\mu\nu,a|x} &=p(+X,a|\mu,\nu,x)-p(-X,a|\mu,\nu,x),
\label{eq:real_amplitude}\\
\operatorname{Im}I_{\mu\nu,a|x} &=p(+Y,a|\mu,\nu,x)-p(-Y,a|\mu,\nu,x).
\label{eq:imag_amplitude}
\end{align}
The derivation is included in Appendix~\ref{app:probe}.

\subsection{Exact one-step realizability}

\begin{theorem}[Faithful-input one-step theorem]
Let $\rho_A\succ0$ be fixed.  Define
\begin{equation}
\widehat C_{a|x} = \left[(\rho_A^{-1}\otimes I_B)\Sigma_{a|x}\right]^{T_A}.
\label{eq:inverse_choi}
\end{equation}
The assemblage is generated by completely positive instruments with one common trace-preserving parent channel if and only if
\begin{align}
\widehat C_{a|x}&\succeq0,
\label{eq:fullrank_pos}\\
\sum_a\widehat C_{a|x}&=C \quad\text{independent of }x,
\label{eq:fullrank_sum}\\
\operatorname{Tr}_B C&=I_A.
\label{eq:fullrank_tp}
\end{align}
The realizing branch maps are unique.
\end{theorem}

\begin{proof}
Because $\rho_A\succ0$, left multiplication by $\rho_A\otimes I_B$ is invertible.  Applying the proposed inverse to Eq.~\eqref{eq:one_step_branch} gives
\begin{align}
\left[ (\rho_A^{-1}\otimes I_B)\Sigma_{a|x} \right]^{T_A} &= \left[ (\rho_A^{-1}\otimes I_B) (\rho_A\otimes I_B)C_{a|x}^{T_A} \right]^{T_A} \nonumber\\
&=C_{a|x},
\label{eq:faithful_input_inverse_check}
\end{align}
so Eq.~\eqref{eq:inverse_choi} is indeed the unique inverse Choi operator. If the family comes from instruments with a common parent, the Choi theorem gives $\widehat C_{a|x}=C_{a|x}\succeq0$, summing the branch maps gives one $x$-independent Choi operator $C$, and trace preservation gives $\operatorname{Tr}_B C=I_A$.  This proves necessity of Eqs.~\eqref{eq:fullrank_pos}--\eqref{eq:fullrank_tp}.

Conversely, Eq.~\eqref{eq:fullrank_pos} and the Choi theorem define a completely positive map for every $a,x$.  Equation~\eqref{eq:fullrank_sum} makes their sum the same map for every $x$, and Eq.~\eqref{eq:fullrank_tp} makes that common map trace preserving.  Finally, transposing Eq.~\eqref{eq:inverse_choi} on $A$ and multiplying by $\rho_A\otimes I_B$ gives
\begin{equation}
(\rho_A\otimes I_B)\widehat C_{a|x}^{T_A} =\Sigma_{a|x},
\end{equation}
so these branch maps reproduce the prescribed QSOST operators.  The inverse calculation in Eq.~\eqref{eq:faithful_input_inverse_check} also proves that no different realizing branch maps are possible.
\end{proof}

This theorem is an affine equivalence between faithful-input left-QSOST assemblages and common-parent channel assemblages.  Its significance is representational and operational: positivity belongs to the inverse Choi image, not to the generally non-Hermitian QSOST operator itself.

For a rank-deficient input, the exact condition is the semidefinite feasibility problem
\begin{align}
C_{a|x}&\succeq0,
\label{eq:rankdef_pos}\\
\sum_a C_{a|x}&=C \quad\text{for every }x,
\label{eq:rankdef_sum}\\
\operatorname{Tr}_B C&=I_A,
\label{eq:rankdef_tp}\\
\Sigma_{a|x}&=(\rho_A\otimes I_B)C_{a|x}^{T_A}.
\label{eq:rankdef_image}
\end{align}
The nonuniqueness of the positive completion outside $\operatorname{supp}\rho_A$ is the one-step precursor of the process-fiber phenomenon analyzed in the main text.

\section{Branch-resolved interferometric reconstruction}
\label{app:readout_proof}

\begin{proposition}[Branch-resolved process readout]
\label{prop:branch_process_readout}
Let \(G_{a|x}\succeq0\) be the unnormalized process branch obtained by conditioning a valid extension on outcome \(a\) of a terminal measurement with setting \(x\).  Assume that the terminal record POVM and the common-future probe measurement act on distinct systems as a product readout, with no feed-forward or joint adaptive response between them.  The record effect is fixed independently of the probe quadrature and of the local-unitary labels used to reconstruct the same branch.  Given a trusted factorwise input--output identification, trusted controlled local unitaries, and a stable phase reference, the interference coefficient isolated by the unnormalized joint probe-record differences for local unitaries \(U_\mu^A,W_\nu^B\) is
\begin{equation}
I_{\mu\nu,a|x} = \operatorname{Tr} \left[ (U_\mu^A\otimes W_\nu^B)Q_+(G_{a|x}) \right].
\label{eq:general_process_interference}
\end{equation}
Consequently, the joint probe-record differences in Eqs.~\eqref{eq:real_amplitude} and \eqref{eq:imag_amplitude} reconstruct \(\Sigma_{a|x}=Q_+(G_{a|x})\) for every terminally conditioned bipartite branch satisfying Eq.~(2) of the main text, not only for causally ordered channels.  The reconstruction uses the unnormalized joint differences; it does not postselect and renormalize the data before forming the branch coefficient.
\end{proposition}

\begin{proof}[Proof of Proposition~\ref{prop:branch_process_readout}]
Write
\begin{equation}
\overline G_{a|x}=\Pi G_{a|x}\Pi^\dagger, \qquad \mathcal G_{a|x}=\overline G_{a|x}^{T_I},
\end{equation}
where the partial transpose is taken in the input basis used for the Choi operators and the swap $S_{I:O}$.  Let $\overline\Upsilon$ be the standard-positive extension in the grouped order $I\otimes O\otimes C_I$, and let $\{E_{a|x}^{C_I}\}_a$ be the terminal POVM.  Partial evaluation of the generalized Born rule defines
\begin{equation}
\overline G_{a|x} = \operatorname{Tr}_{C_I} \left[ (I_{IO}\otimes E_{a|x})\overline\Upsilon \right].
\label{eq:terminal_branch_grouped}
\end{equation}
This operator is positive because the partial trace permits the equivalent sandwiched expression
\begin{align}
\overline G_{a|x} ={}& \operatorname{Tr}_{C_I} \left[ (I_{IO}\otimes E_{a|x}^{1/2}) \overline\Upsilon (I_{IO}\otimes E_{a|x}^{1/2}) \right] \succeq0.
\end{align}

Let $\widetilde I$ be a reference copy of $I$, let $S_{I:\widetilde I}$ be the swap between these two spaces, and let $J:\widetilde I\rightarrow O$ be the fixed factorwise unitary identification used throughout.  Thus
\begin{equation}
S_{I:O} = (I_I\otimes J)S_{I:\widetilde I}(I_I\otimes J^\dagger).
\end{equation}
For an effective completely positive event $\mathcal M:\mathcal L(\widetilde I)\rightarrow\mathcal L(O)$ realized by the local-probe construction specified below, introduce the swap-Jamio\l kowski representative
\begin{equation}
J_0[\mathcal M] = (\operatorname{id}_I\otimes\mathcal M)(S_{I:\widetilde I}).
\label{eq:swap_jamiolkowski}
\end{equation}
Since
\begin{equation}
S_{I:\widetilde I}^{T_I} = |\Omega_{I:\widetilde I}\rangle \langle\Omega_{I:\widetilde I}|,
\end{equation}
the standard positive Choi operator used in Eq.~\eqref{eq:choi} obeys
\begin{equation}
C[\mathcal M]=J_0[\mathcal M]^{T_I}.
\label{eq:two_choi_conventions}
\end{equation}
Here $C[\mathcal M]$ abbreviates the contraction of the explicit entangled-probe and locally controlled circuit constructed below.  It does not assume access to an arbitrary nonlocal operation inside the laboratories. For these physically implemented effective events, contraction of the enlarged local process, the probe state, and the common-future probe measurement gives
\begin{align}
p(\mathcal M,a|x) &= \operatorname{Tr} \left[ (C[\mathcal M]\otimes E_{a|x})\overline\Upsilon \right] \nonumber\\
&= \operatorname{Tr} \left[ C[\mathcal M]\overline G_{a|x} \right] \nonumber\\
&= \operatorname{Tr} \left[ J_0[\mathcal M]\mathcal G_{a|x} \right].
\label{eq:conditioned_generalized_born}
\end{align}

Set
\begin{equation}
V_{\mu\nu}=U_\mu^A\otimes W_\nu^B,
\end{equation}
as a unitary on $O$.  For $b\in\{+1,-1\}$, define the two probe-readout instruments by the single Kraus operators
\begin{align}
K_{b,X} &= \frac{J+bV_{\mu\nu}J}{2}, \\
K_{b,Y} &= \frac{J-ibV_{\mu\nu}J}{2},
\end{align}
and $\mathcal M_{b,\xi}(X)=K_{b,\xi}XK_{b,\xi}^\dagger$ for $\xi\in\{X,Y\}$.  Unitarity of $V_{\mu\nu}$ gives
\begin{equation}
\sum_{b=\pm1}K_{b,\xi}^\dagger K_{b,\xi} =J^\dagger J=I_{\widetilde I},
\end{equation}
so each pair is a valid two-outcome quantum instrument.

Define
\begin{align}
z_{\mu\nu,a|x} &= \operatorname{Tr} \left[ (I_I\otimes V_{\mu\nu}) S_{I:O}\mathcal G_{a|x} \right], \\
r_{0,a|x} &= \operatorname{Tr} \left[ S_{I:O}\mathcal G_{a|x} \right], \\
r_{1,\mu\nu,a|x} &= \operatorname{Tr} \left[ (I_I\otimes V_{\mu\nu}) S_{I:O} (I_I\otimes V_{\mu\nu}^\dagger) \mathcal G_{a|x} \right].
\end{align}
Because $\mathcal G_{a|x}=\mathcal G_{a|x}^\dagger$ and $S_{I:O}=S_{I:O}^\dagger$,
\begin{equation}
\operatorname{Tr} \left[ S_{I:O}(I_I\otimes V_{\mu\nu}^\dagger) \mathcal G_{a|x} \right] = z_{\mu\nu,a|x}^*.
\end{equation}
Expanding Eq.~\eqref{eq:conditioned_generalized_born} yields
\begin{align}
p(b,X,a|\mu,\nu,x) &= \frac14 \left[ r_{0,a|x}+r_{1,\mu\nu,a|x} +2b\,\operatorname{Re}z_{\mu\nu,a|x} \right], \\
p(b,Y,a|\mu,\nu,x) &= \frac14 \left[ r_{0,a|x}+r_{1,\mu\nu,a|x} +2b\,\operatorname{Im}z_{\mu\nu,a|x} \right].
\end{align}
The diagonal terms need not coincide for a fixed terminal outcome, but they cancel in the joint outcome differences.  Hence
\begin{align}
p(+X,a|\mu,\nu,x)-p(-X,a|\mu,\nu,x) &= \operatorname{Re}z_{\mu\nu,a|x}, \\
p(+Y,a|\mu,\nu,x)-p(-Y,a|\mu,\nu,x) &= \operatorname{Im}z_{\mu\nu,a|x}.
\end{align}
Finally,
\begin{align}
z_{\mu\nu,a|x} &= \operatorname{Tr} \left[ V_{\mu\nu} \operatorname{Tr}_I \left( S_{I:O}\mathcal G_{a|x} \right) \right] \nonumber\\
&= \operatorname{Tr} \left[ (U_\mu^A\otimes W_\nu^B)Q_+(G_{a|x}) \right].
\end{align}
Thus the joint, unnormalized probe-record differences recover both quadratures of the branch coefficient, and Eq.~\eqref{eq:tomography} reconstructs $Q_+(G_{a|x})$.
\end{proof}

The effective instruments above can be implemented without assuming an arbitrary nonlocal operation inside the laboratories.  Prepare the two probes in
\begin{equation}
|+_L\rangle = \frac{|00\rangle+|11\rangle}{\sqrt2},
\end{equation}
apply the locally controlled $U_\mu^A$ and $W_\nu^B$, and measure the probes jointly in their common future in the logical bases
\begin{align}
|b_X\rangle_L &= \frac{|00\rangle+b|11\rangle}{\sqrt2}, \\
|b_Y\rangle_L &= \frac{|00\rangle+ib|11\rangle}{\sqrt2}, \qquad b\in\{+1,-1\}.
\end{align}
Projection onto these states gives $K_{b,X}=(J+bV_{\mu\nu}J)/2$ and $K_{b,Y}=(J-ibV_{\mu\nu}J)/2$, respectively.  The readout therefore assumes trusted controlled unitaries, a stable phase reference, and a common-future probe measurement, but no prior knowledge of the causal order between $A$ and $B$.

\section{Fully specified terminal TTU branch families}
\label{app:terminal_structure}

\begin{proposition}[Terminal TTU realization theorem]
\label{prop:future_realization}
A fully specified finite family $\{G_{a|x}\}_{a,x}$ is generated by POVMs on a party $C$ in the common future of $A$ and $B$ if and only if
\begin{equation}
G_{a|x}\succeq0, \qquad \sum_aG_{a|x}=W\in\mathfrak P \quad\text{for every }x,
\label{eq:future_realization_condition}
\end{equation}
where $C$ has a trivial output space.
\end{proposition}

\begin{proof}[Proof of Proposition~\ref{prop:future_realization}]
For necessity, let $\Upsilon^{A_IA_OB_IB_OC_I}$ be the tripartite process and let $\{E_{a|x}^{C_I}\}_a$ be Charlie's POVM.  Then
\begin{equation}
G_{a|x} = \operatorname{Tr}_{C_I} \left[ (I\otimes E_{a|x})\Upsilon \right].
\label{eq:future_branch_contraction}
\end{equation}
Positivity is immediate, and POVM completeness gives
\begin{equation}
\sum_aG_{a|x} = \operatorname{Tr}_{C_I}\Upsilon =W,
\end{equation}
independently of $x$.  Because $C$ is terminal, this marginal is a deterministic bipartite process.

For sufficiency, diagonalize $W$ on its support,
\begin{equation}
W=\sum_{k=1}^{r}\lambda_k|k\rangle\langle k|, \qquad \lambda_k>0,
\end{equation}
and introduce an $r$-dimensional terminal input $C_I$.  Write $P=A_I A_O B_I B_O$ for the process Hilbert space and define
\begin{equation}
|\Psi_W\rangle = \sum_{k=1}^{r}\sqrt{\lambda_k} |k\rangle_{P}|k\rangle_{C_I}, \qquad \Upsilon=|\Psi_W\rangle\langle\Psi_W|.
\label{eq:canonical_future_extension}
\end{equation}
The operator $\Upsilon$ is positive and $\operatorname{Tr}_{C_I}\Upsilon=W$.  Because $C$ has no output, for every pair of deterministic local operations $M_A,M_B$,
\begin{equation}
\operatorname{Tr}[(M_A\otimes M_B\otimes I_{C_I})\Upsilon] = \operatorname{Tr}[(M_A\otimes M_B)W] =1,
\end{equation}
so $\Upsilon$ is a valid tripartite process with $C$ in the future. Moreover, Eq.~\eqref{eq:future_realization_condition} implies $0\preceq G_{a|x}\preceq W$ and hence $\operatorname{supp}G_{a|x}\subseteq\operatorname{supp}W$.  In the displayed eigenbasis set
\begin{equation}
E_{a|x} = \left( W^{-1/2}G_{a|x}W^{-1/2} \right)^T,
\label{eq:future_povm_general}
\end{equation}
where the inverse is the Moore--Penrose inverse on $\operatorname{supp}W$. Congruence and transposition preserve positivity, the support inclusion ensures exact recovery, and Eq.~\eqref{eq:future_realization_condition} gives
\begin{equation}
\sum_aE_{a|x}=I_{C_I}.
\end{equation}
Finally, direct contraction of Eq.~\eqref{eq:canonical_future_extension} gives
\begin{equation}
\operatorname{Tr}_{C_I} \left[ (I\otimes E_{a|x})\Upsilon \right] = W^{1/2}E_{a|x}^{T}W^{1/2} =G_{a|x}.
\end{equation}
\end{proof}

This is precisely the terminal TTU scenario characterized in Ref.~[S2].  Its constructive step is the unnormalized operator form of the S-G-HJW argument ~[S3]; compare Ref.~[S4]. Theorem 5 of the latter reference includes process matrices as quantum objects when every assemblage element obeys the corresponding structural projector.  General terminal TTU branches used here need not obey $L_V(G_{a|x})=G_{a|x}$.  The exact fully specified precedent is therefore Theorem 9 of Ref.~[S2], rather than a direct application of the arbitrary-object theorem.

Imposing this process-subspace condition on every positive branch would restrict it to a subnormalized deterministic process, whose occurrence probability is independent of the local deterministic interventions.

\section{Boundary with simultaneous purification and marginal problems}
\label{app:marginal_boundary}

Simultaneous-purification theorems start from a fully specified non-signalling assemblage of valid quantum objects,
\begin{equation}
\sum_a O_{a|x}=O \quad\text{for every }x,
\label{eq:full_object_nonsignalling}
\end{equation}
and construct one dilation followed by delayed measurements ~[S4].  Their state, instrument, and super-instrument inputs are positive objects satisfying the corresponding normalization and structural constraints.  A terminal process branch in the present convention is positive, but need not satisfy $L_V(G_{a|x})=G_{a|x}$; its QSOST image can moreover be non-Hermitian. Consequently, neither the full terminal branch family nor its compressed image is automatically an assemblage of the type assumed in those theorems.

\begin{proposition}[Exact boundary with simultaneous purification]
\label{prop:simultaneous_purification_boundary}
Let a finite compressed family obey
\begin{equation}
\sum_a\Sigma_{a|x}=T \quad\text{for every }x.
\end{equation}
The following are distinct statements:
\begin{align}
&\exists\,\{G_{a|x}\}_{a,x},\{W_x\}_x: \nonumber\\[-2pt]
&\quad G_{a|x}\succeq0,\qquad Q_+(G_{a|x})=\Sigma_{a|x} \quad\forall a,x, \nonumber\\[-2pt]
&\quad \sum_aG_{a|x}=W_x\in\mathfrak P \quad\forall x,
\label{eq:boundary_settingwise_full}\\
&\exists\,\{G_{a|x}\}_{a,x},W: \nonumber\\[-2pt]
&\quad G_{a|x}\succeq0,\qquad Q_+(G_{a|x})=\Sigma_{a|x} \quad\forall a,x, \nonumber\\[-2pt]
&\quad \sum_aG_{a|x}=W\in\mathfrak P \quad\forall x.
\label{eq:boundary_common_full}
\end{align}
The second statement implies the first and admits one common-future process extension with a terminal POVM.  Among families satisfying the first statement, the implication from the first to the second holds for every finite compressed family if and only if $Q_+$ is injective on $\mathfrak P$.  Hence it fails in every nontrivial equal-local-dimension bipartite scenario.
\end{proposition}

\begin{proof}
If Eq.~\eqref{eq:boundary_common_full} holds, set $W_x=W$ for every $x$; all conditions in Eq.~\eqref{eq:boundary_settingwise_full} then hold, so the second statement implies the first.  The same hypothesis supplies a fully specified finite family of positive branches whose sums equal one deterministic process.  These are exactly the hypotheses of the terminal-party realization theorem, Proposition~\ref{prop:future_realization}, which therefore produces one common-future process extension and terminal POVMs realizing all the branches.

Now restrict attention to compressed families satisfying the first statement.  They are precisely the settingwise-realizable families in Theorem~1 of the main text: their branch images sum to the same $T$, but their deterministic lifts may be the different $W_x$.  The second statement is precisely common realizability for those data.  Hence the implication from the first statement to the second holds for every such finite family if and only if $Q_+|_{\mathfrak P}$ is injective, by Theorem~1 of the main text.  For $m,n\ge2$ the hidden-fiber corollary, Corollary~1 of the main text, gives distinct same-compression deterministic processes and therefore a family for which the implication fails.
\end{proof}

Thus there is no conflict with simultaneous purification.  Equality in Eq.~\eqref{eq:full_object_nonsignalling} is imposed before dilation, whereas the present data impose equality only after the non-injective map $Q_+$:
\begin{equation}
Q_+(W_x)=T \quad\text{does not imply}\quad W_x=W.
\label{eq:compressed_vs_full_nonsignalling}
\end{equation}
The compression--gluing theorem characterizes exactly when that missing implication can nevertheless be recovered.

The quantum channel marginal problem of Ref.~[S5] starts from CPTP maps assigned to different, possibly overlapping input--output routes and asks for a global CPTP map whose Choi operator satisfies
\begin{equation}
\operatorname{Tr}_{S\setminus X} C[\mathcal E_{S|S_0}] = C[\mathcal E_{X|X_0}] \otimes I_{S_0\setminus X_0}
\end{equation}
for every specified route $X_0\to X$ in the unnormalized Choi convention of Eq.~\eqref{eq:choi}.  Equivalently, for normalized Choi states the last factor is $I_{S_0\setminus X_0}/d_{S_0\setminus X_0}$.  These constraints combine an output-subsystem reduction with the quantum no-signalling condition on the omitted inputs.  By contrast, all $\Sigma_{a|x}$ considered here refer to the same laboratory slots, $x$ indexes alternative branch decompositions rather than different routes, and $Q_+$ is not this reduced-channel map.  The present problem is therefore distinct from the channel marginal problem as defined in Ref.~[S5], although both are conic feasibility problems under linear information constraints.

The distinction is also algebraic, not merely terminological.  A subsystem partial trace is completely positive and maps every positive operator to a positive Hermitian operator.  The QSOST contraction does not.  Indeed, Eq.~\eqref{eq:Q_product} gives
\begin{equation}
Q_+(X_I\otimes Y_O)=X_I^T Y_O,
\end{equation}
which is non-Hermitian whenever $X_I^T$ and $Y_O$ are positive and do not commute.  More explicitly, for
\begin{equation}
0<\epsilon \le \frac{1}{d_O\|X_I\otimes Y_O\|_\infty},
\end{equation}
the two positive branches
\begin{equation}
G_0=\epsilon X_I\otimes Y_O, \qquad G_1=\frac{I}{d_O}-G_0
\end{equation}
sum to the full-rank deterministic process $I/d_O$.  They therefore form a terminal assemblage by Proposition~\ref{prop:future_realization}, while $Q_+(G_0)$ is non-Hermitian.  Thus even on physical branches $Q_+$ cannot be converted into a channel marginal map by tensor relabeling or unitary conjugation.

At faithful input, the one-step theorem above is likewise an equivalence with an instrument assemblage on one fixed input--output route, not a channel marginal problem across several routes.  Its invertibility is precisely why the hidden positive-lift obstruction central to the main theorem does not occur there.

Nor is the present task the standard party- or slot-marginal consistency problem for process matrices or quantum combs.  Such problems specify reduced higher-order operators on different parties, routes, or time slots and impose the corresponding linking and normalization constraints.  Here every datum refers to the same complete $A,B$ slot set, $x$ labels alternative decompositions, and $Q_+$ is neither a slot reduction nor a link product ~[S6, S7].

The quantum space-time marginal problem of Ref.~[S8] starts from compatible pseudo-density operators $R_{A_i}$ on different event subsets and asks for a global pseudo-density operator $R_A$ satisfying
\begin{equation}
R_{A_i} = \operatorname{Tr}_{A\setminus A_i}R_A \quad\text{for every }i.
\end{equation}
Its unconstrained linear relaxation always has an extension in the space $\operatorname{Herm}_1(A)$ of Hermitian trace-one operators, whereas positivity, separability, or physical realizability impose additional restrictions.  Here no larger event set is introduced and no event subsystem is traced away.  Instead, one searches inside the standard positive process cone for preimages of a fixed non-injective contraction and then imposes one common deterministic-process normalization.  The shared feature is completion under linear information loss, not a common marginal map or an equivalence of the underlying formalisms.  In particular, the three inputs have different types: reduced CPTP Choi operators on overlapping routes in the channel marginal problem, Hermitian pseudo-density operators on event subsets in the space-time marginal problem, and generally non-Hermitian images of unnormalized terminal branches on one fixed set of slots here.

\section{Dual process-domination program and generic witnesses}
\label{app:dual_witness}

The dual SDP is
\begin{equation}
\mu(\{R_x\}) = \max_{\{Y_x\}} \sum_x\operatorname{Tr}(Y_xR_x)
\label{eq:mu_dual_objective}
\end{equation}
subject to
\begin{equation}
Y_x\succeq0, \qquad L_V\!\left(\sum_xY_x\right)=I.
\label{eq:mu_dual_constraints}
\end{equation}
For completeness, the dual follows directly by minimizing the Lagrangian over the real vector space $\operatorname{ran}L_V\cap\operatorname{Herm}$. Attach $Y_x\succeq0$ to each constraint $W-R_x\succeq0$.  The Lagrangian is
\begin{align}
\mathcal L(W,\{Y_x\}) &= \operatorname{Tr}W -\sum_x\operatorname{Tr}[Y_x(W-R_x)] \nonumber\\
&= \sum_x\operatorname{Tr}(Y_xR_x) + \operatorname{Tr}\!\bigl[ \left(I-\sum_xY_x\right)W \bigr].
\label{eq:mu_lagrangian}
\end{align}
Because $W=L_V(W)$ and $L_V$ is self-adjoint, the infimum over this unconstrained linear space is finite exactly when
\begin{equation}
L_V\!\left(I-\sum_xY_x\right)=0.
\end{equation}
Since $L_V(I)=I$, this is precisely Eq.~\eqref{eq:mu_dual_constraints}; in that case the infimum equals the objective in Eq.~\eqref{eq:mu_dual_objective}.

The primal is strictly feasible.  The setting set is finite, so choosing
\begin{equation}
t>\max_x\|R_x\|_\infty
\end{equation}
and $W_0=tI$ gives $L_V(W_0)=W_0$ and $W_0-R_x\succ0$ for every $x$.  The primal also has a finite optimum, as shown below Eq.~(17) of the main text.  Slater's condition therefore gives equality and attainment of the primal and dual optimal values.  Hence every common-process assemblage obeys the QSOST-only nonlinear witness
\begin{equation}
\sum_{x,a} \frac{\langle q_{a|x}|Y_x|q_{a|x}\rangle}{p_{a|x}} \le d_O
\label{eq:nonlinear_witness}
\end{equation}
for every dual-feasible $\{Y_x\}$.  Zero branches have $p_{a|x}=0$ and $|q_{a|x}\rangle=0$ and are omitted from the sum, equivalently contributing zero.

A linear interferometric witness follows by supporting-hyperplane linearization.  Fix a reference assemblage with $p^0_{a|x}>0$ and vectors $|q^0_{a|x}\rangle$.  The quadratic-over-linear inequality
\begin{align}
\frac{\langle q|Y|q\rangle}{p} \ge{}& \frac{2\operatorname{Re}\langle q^0|Y|q\rangle}{p^0} - \frac{\langle q^0|Y|q^0\rangle}{(p^0)^2}p
\label{eq:quadratic_tangent}
\end{align}
holds for $Y\succeq0$, because their difference is
\begin{equation}
\frac{1}{p} \left\|Y^{1/2}|q\rangle- \frac{p}{p^0}Y^{1/2}|q^0\rangle \right\|^2.
\end{equation}
Combining Eqs.~\eqref{eq:nonlinear_witness} and \eqref{eq:quadratic_tangent} gives the linear inequality
\begin{align}
\sum_{x,a} \left[ \frac{2\operatorname{Re}\langle q^0_{a|x}|Y_x|q_{a|x}\rangle}{p^0_{a|x}} - \frac{\langle q^0_{a|x}|Y_x|q^0_{a|x}\rangle}{(p^0_{a|x})^2} p_{a|x} \right] \le d_O.
\label{eq:linear_witness}
\end{align}
The left-hand side is real-linear in the QSOST operators and can be expanded in the measured unitary basis of Eq.~\eqref{eq:interference_coefficients}.

\section{Causal location and accessible actual records}
\label{app:causal_records}
\subsection{Causal location of the setting label}

The preceding distinction has an exact realization-theoretic causal boundary.  A \emph{future-choice realization} places the setting $x$ at a terminal party after the laboratories, assumes that it is generated independently of the process source, and lets it control only the terminal POVM.  There is no setting leakage or feed-forward to the source or to the earlier interferometric interfaces.  A \emph{past-controlled realization} instead places $x$ in a classical register before the laboratories and permits it to select their process context.

\begin{theorem}[Future choice versus past control]
\label{thm:causal_location}
For any finite family of fully specified positive branches $\{G_{a|x}\}_{a,x}$:
\begin{enumerate}
\item it has a future-choice realization if and only if it has one common deterministic parent;
\item it has a past-controlled realization if and only if every setting has a deterministic parent $W_x$.
\end{enumerate}
\end{theorem}

\begin{proof}
The first statement is Proposition~\ref{prop:future_realization}.  For the nontrivial direction of the second statement, assume
\begin{equation}
G_{a|x}\succeq0, \qquad \sum_aG_{a|x}=W_x\in\mathfrak P.
\label{eq:settingwise_branches_causal}
\end{equation}
Add zero outcomes if necessary so that all settings have one outcome alphabet.  Let $X$ be a classical system with orthonormal basis $\{|x\rangle\}_x$ and define
\begin{align}
\Theta^{X A_IA_OB_IB_O} &= \sum_x|x\rangle\langle x|^X\otimes W_x,
\label{eq:controlled_parent_process}\\
\Gamma_a^{X A_IA_OB_IB_O} &= \sum_x|x\rangle\langle x|^X\otimes G_{a|x}.
\label{eq:controlled_branch_process}
\end{align}
The operator $\Theta$ is a valid process with $X$ in the causal past.  Indeed, it is positive and, for every state $\rho_X$ prepared by the past controller and every pair of deterministic local operations with Choi operators $M_A,M_B$,
\begin{align}
&\operatorname{Tr} \left[ (\rho_X\otimes M_A\otimes M_B)\Theta \right] \nonumber\\
&\quad= \sum_x\langle x|\rho_X|x\rangle \operatorname{Tr}[(M_A\otimes M_B)W_x] =1.
\label{eq:controlled_process_normalization}
\end{align}
The block-diagonal form also preserves positivity under arbitrary ancillary extensions.  Moreover,
\begin{equation}
\Gamma_a\succeq0, \qquad \sum_a\Gamma_a=\Theta.
\end{equation}
The HJW construction in Proposition~\ref{prop:future_realization} is unchanged for any finite collection of nonterminal slots: purify the full deterministic parent into one terminal input and use the transposed Radon--Nikodym effects on its support.  Its normalization proof contracts the terminal identity back to that parent.  Applying this construction to the enlarged parent $\Theta$ produces one terminal record system and one fixed POVM $\{F_a\}_a$ whose branch processes are $\Gamma_a$.  Preparing $|x\rangle\langle x|$ at the past controller gives
\begin{equation}
\operatorname{Tr}_X \left[ (|x\rangle\langle x|\otimes I)\Gamma_a \right] =G_{a|x}.
\end{equation}
Thus one past-controlled process realizes the entire settingwise family.  The converse follows by conditioning any such controlled process on the classical input $x$.
\end{proof}

Let $\mathcal C_{\mathrm{future}}$ and $\mathcal C_{\mathrm{past}}$ denote the corresponding classes after applying $Q_+$ to the branch processes, restricted throughout to compressed families obeying Eq.~(12) of the main text.  The theorem gives
\begin{equation}
\mathcal C_{\mathrm{future}} = \mathcal C_{\mathrm{common}}, \qquad \mathcal C_{\mathrm{past}} = \mathcal C_{\mathrm{settingwise}}.
\label{eq:causal_class_identification}
\end{equation}
Combining Eq.~\eqref{eq:causal_class_identification} with the compression--gluing theorem yields the operational corollary
\begin{equation}
\boxed{ Q_+|_{\mathfrak P}\text{ is injective} \quad\Longleftrightarrow\quad \mathcal C_{\mathrm{future}} = \mathcal C_{\mathrm{past}}.}
\label{eq:causal_compression_equivalence}
\end{equation}
Hence noninjectivity is exactly the condition under which some compressed families are compatible with past control but incompatible with fixed-slot future choice.  This feasibility separation does not identify the mechanism that generated a given data set.

\begin{corollary}[Trusted fixed-slot future-choice test]
\label{prop:trusted_causal_location}
Assume the source is independent of the early interferometric label.  After $A,B$ close, let fresh randomness generate $x$ independently of source-side variables and earlier device states, with $x$ controlling only a terminal POVM on one fixed process extension.  Assume the calibrated QSOST interface and an isolated, factorized probe--record readout, without leakage, feed-forward, postselection, or setting-dependent recalibration.  Every such model obeys, for the binary score above,
\begin{equation}
\mu(\{R_x\})=d_O, \qquad \mathcal W\le\sqrt2.
\label{eq:trusted_future_null}
\end{equation}
Consequently, $\mu>d_O$ or $\mathcal W>\sqrt2$ rules out the fixed-slot future-choice model under these assumptions.
\end{corollary}

\begin{proof}
The timing assumption places the fresh variable $x$ after the declared $A,B$ slots, while source independence, absence of leakage and feed-forward, and the fixed calibrated, factorized probe--record interface ensure that $x$ changes only the terminal POVM.  Thus the fully specified positive branches have a future-choice realization in the sense of Theorem~\ref{thm:causal_location}.  That theorem gives a single $x$-independent deterministic parent $W$.  Applying the process-domination criterion to its compressed branches yields $W\succeq R_x$ for every $x$; since $\operatorname{Tr}W=d_O$, Eqs.~(16) and (17) of the main text give $\mu(\{R_x\})=d_O$.  The same common-parent branches satisfy the calibrated sparse witness bound Eq.~(31) of the main text, and therefore $\mathcal W\le\sqrt2$.  These are exactly the two assertions in Eq.~\eqref{eq:trusted_future_null}; their contrapositives give the stated exclusions.
\end{proof}

\subsection{Accessible actual records; no delayed-fact theorem}

The causal-location theorem concerns the setting actually chosen and the outcome actually obtained.  This operational content can be made precise without assuming simultaneous values for measurements that were not performed.  The results in this subsection are a logically independent boundary analysis: they neither follow from compression noninjectivity alone nor supply a premise for Theorem~1 of the main text.

\begin{corollary}[Perfectly accessible actual transcript]
\label{cor:accessible_actual_record}
Let $\{G_{a|x}\}_{a,x}$ be a finite family of positive branch operators. Call its actual transcript perfectly accessible if the realization can be augmented by a terminal classical register $R$ whose orthogonal states $|x,a\rangle_R$ record the setting actually used and the outcome actually obtained.  Then:
\begin{enumerate}
\item such a future-choice realization exists if and only if the branches
have one common deterministic parent;
\item such a past-controlled realization exists if and only if every setting
has a deterministic parent.
\end{enumerate}
Thus adding a perfectly accessible record of the actual transcript does not change either realizability class.
\end{corollary}

\begin{proof}
Necessity follows by discarding $R$ and applying Theorem~\ref{thm:causal_location}.  For sufficiency in the future-choice case, use the extension $\Upsilon$ and POVMs $\{E_{a|x}\}_a$ from Proposition~\ref{prop:future_realization}, and replace the terminal POVM by the quantum-to-classical instrument
\begin{equation}
\mathcal J_{a|x}(\tau) = \operatorname{Tr}(E_{a|x}\tau) |x,a\rangle\langle x,a|^R .
\label{eq:actual_record_instrument}
\end{equation}
Conditioning on the pointer value gives $G_{a|x}\otimes|x,a\rangle\langle x,a|^R$.

For the past-controlled case, use the block-diagonal parent $\Theta$ and the fixed terminal POVM $\{F_a\}_a$ from Theorem~\ref{thm:causal_location}.  Since $x$ is a classical input, it can be copied to a forward classical memory, while the outcome of $\{F_a\}_a$ is written into the same terminal register.  This produces the same orthogonal pointer without changing the branch operators.  This classical memory wire from the setting controller to the terminal register is part of the allowed augmentation; it does not alter the process on the original $A,B$ slots. Explicitly, before recording $a$ one may use the extension
\begin{equation}
\widetilde\Theta = \sum_x |x\rangle\langle x|^X \otimes W_x \otimes|x\rangle\langle x|^{R_X}, \qquad \operatorname{Tr}_{R_X}\widetilde\Theta=\Theta.
\end{equation}

For any distribution $p(x,a)$ of actual transcripts, the resulting classical--quantum state is
\begin{equation}
\rho_{ZR} = \sum_{x,a}p(x,a) |x,a\rangle\langle x,a|^{Z} \otimes |x,a\rangle\langle x,a|^R,
\end{equation}
and therefore $H(Z|R)_\rho=0$.  Moreover, the accessible algebra
\begin{equation}
\mathcal A_R = \operatorname{span} \left\{ |x,a\rangle\langle x,a|^R \right\}_{x,a}
\end{equation}
is commutative.  It encodes the one pair $(x,a)$ that occurred, not a tuple of outcomes for all possible settings.  Every record state diagonal in this actual-record algebra can be broadcast to any finite number of readers through
\begin{equation}
\mathcal B_m(\tau) = \sum_{x,a} \langle x,a|\tau|x,a\rangle \left( |x,a\rangle\langle x,a| \right)^{\otimes m}.
\label{eq:record_broadcast_channel}
\end{equation}
\end{proof}

\begin{proposition}[Common gluing need not admit a joint operator-valued refinement]
\label{prop:no_counterfactual_record}
By a setting-independent classical refinement we mean precisely positive operators $\{H_\lambda\}_\lambda$ and response probabilities $p(a|x,\lambda)$ with $p(a|x,\lambda)\ge0$ and $\sum_a p(a|x,\lambda)=1$ satisfying
\begin{align}
\sum_\lambda H_\lambda&=W,\\
G_{a|x}&=\sum_\lambda p(a|x,\lambda)H_\lambda.
\label{eq:counterfactual_refinement}
\end{align}
Common-process realizability does not in general imply that such a refinement exists.
\end{proposition}

\begin{proof}
Choose a full-rank deterministic process $W$ and two incompatible POVMs $\{E_{a|0}\}_a$ and $\{E_{a|1}\}_a$ on the process Hilbert space, with the transpose taken in the eigenbasis of $W$.  For example, one may use two noncommuting sharp binary measurements.  Define
\begin{equation}
G_{a|x} = W^{1/2}E_{a|x}^{T}W^{1/2}.
\end{equation}
These branches are positive and sum to $W$ for every $x$, so they have a common-process realization.

If the refinement in Eq.~\eqref{eq:counterfactual_refinement} existed, then
\begin{equation}
F_\lambda = \left( W^{-1/2}H_\lambda W^{-1/2} \right)^T
\end{equation}
would form a POVM and
\begin{equation}
E_{a|x} = \sum_\lambda p(a|x,\lambda)F_\lambda.
\end{equation}
Both POVMs would therefore be postprocessings of one mother POVM and hence jointly measurable, contradicting their incompatibility.
\end{proof}

Corollary~\ref{cor:accessible_actual_record} gives a precise terminal-record statement: gluing determines whether the setting can be chosen in the common future while retaining a perfectly readable record of the actual event. Proposition~\ref{prop:no_counterfactual_record} rules out the specific stronger inference that the branches must be postprocessings of one classical mother variable.  Hence common gluing alone cannot establish simultaneous values for all settings or an observer-independent counterfactual fact table; this is an operator-level mother-instrument obstruction, not a universal no-go for every possible ontological formalization of observer-independent facts.  The time at which a fact forms, and multi-observer objectivity based on redundant environmental records, require additional dynamical assumptions not contained in process-assemblage data alone.

\begin{corollary}[Quantitative cost of exposing two parents]
For the exposed family generated by two same-fiber processes $W_0,W_1\in\mathfrak P$, every primal-feasible common dominator $W\succeq W_0,W_1$ obeys
\begin{equation}
\operatorname{Tr}W \ge d_O+\frac12\|W_0-W_1\|_1.
\label{eq:generic_exposure_gap}
\end{equation}
\end{corollary}

\begin{proof}
Set $A=W-W_0\succeq0$ and $B=W-W_1\succeq0$.  Since $W_0-W_1=B-A$, the triangle inequality gives
\begin{equation}
\|W_0-W_1\|_1 \le\operatorname{Tr}A+\operatorname{Tr}B =2\operatorname{Tr}W-2d_O.
\end{equation}
\end{proof}

Applying the corollary to the qubit specialization of the full-rank pair $W_\pm$ in Eq.~\eqref{eq:general_hidden_pair} gives a two-setting, sixteen-outcome strict separation with a maximally mixed unconditional QSOST.  This shows that the separation is not confined to a boundary face of the process cone.  The corollary gives
\begin{equation}
\operatorname{Tr}W \ge \frac{\operatorname{Tr}W_++\operatorname{Tr}W_-+\|W_+-W_-\|_1}{2} = 4+16\epsilon.
\label{eq:fullrank_gap}
\end{equation}

\section{Full-rank binary separation under white noise}
\label{app:fullrank_noise}

The minimal construction is not a boundary artifact.  Let $0<\zeta<1$ and add the same full-rank white-noise branch to both outcomes:
\begin{equation}
\Sigma_{\pm|x}^{[\zeta]} = (1-\zeta)\Sigma_{\pm|x} + \frac{\zeta}{8}I_4, \qquad x\in\{\mathrm R,\mathrm D\}.
\label{eq:fullrank_noisy_qsost}
\end{equation}
The two settings retain the same unconditional QSOST,
\begin{equation}
\sum_{\pm}\Sigma_{\pm|x}^{[\zeta]} = (1-\zeta)|00\rangle\langle00| + \frac{\zeta}{4}I_4.
\label{eq:fullrank_noisy_unconditional}
\end{equation}

\begin{theorem}[Full-rank binary gluing separation]
\label{thm:fullrank_binary}
For every
\begin{equation}
0<\zeta<1-\frac1{\sqrt2},
\label{eq:fullrank_noise_range}
\end{equation}
the two-setting, two-outcome family in Eq.~\eqref{eq:fullrank_noisy_qsost} is settingwise realizable by positive-definite branch operators whose deterministic parents are positive definite and have order $A\prec B$, but it has no common deterministic parent.
\end{theorem}

\begin{proof}
Use the branches from Eq.~\eqref{eq:realizing_two_branches} and define
\begin{align}
G_{\pm|x}^{[\zeta]} &= (1-\zeta)G_{\pm|x} + \frac{\zeta}{8}I_{16},
\label{eq:fullrank_noisy_branches}\\
W_x^{[\zeta]} &= \sum_{\pm}G_{\pm|x}^{[\zeta]} = (1-\zeta)W_x + \frac{\zeta}{4}I_{16}.
\label{eq:fullrank_noisy_parents}
\end{align}
Linearity and $Q_+(I_{16}/4)=I_4/4$ show that the branch images are Eq.~\eqref{eq:fullrank_noisy_qsost}.  Moreover,
\begin{equation}
G_{\pm|x}^{[\zeta]}\succeq\frac{\zeta}{8}I_{16}, \qquad W_x^{[\zeta]}\succeq\frac{\zeta}{4}I_{16},
\end{equation}
so both branches and parents are positive definite.  The white-noise process can be written as
\begin{equation}
\frac{I_{16}}4 = \frac{I^{A_I}}2 \otimes \frac{I^{A_OB_I}}2 \otimes I^{B_O},
\end{equation}
where $I^{A_OB_I}/2$ is the Choi operator of the completely depolarizing channel from $A_O$ to $B_I$.  Hence Eq.~\eqref{eq:fullrank_noisy_parents} remains inside the $A\prec B$ process cone.

The equal white-noise terms cancel from every branch difference in Eq.~\eqref{eq:sparse_W_vector}.  Therefore the witness value is
\begin{equation}
\mathcal W^{[\zeta]}=2(1-\zeta).
\end{equation}
The range in Eq.~\eqref{eq:fullrank_noise_range} gives $\mathcal W^{[\zeta]}>\sqrt2$, contradicting the common-process bound in Eq.~(31) of the main text.  Thus no common deterministic parent exists.
\end{proof}

\section{Explicit controlled realization and causal interpretation}
\label{app:operational_details}

The separation has a precise operational meaning.  Under the trusted readout assumptions, causally agnostic interferometric data reconstruct the branch QSOST operators but not the full positive process matrices.  Each record setting can be completed to a valid standard-quantum process, and the unconditional QSOST is identical for all settings.  Nevertheless, the branch data expose incompatible hidden positive completions, so no $x$-independent parent exists on the original $A,B$ process slots. Theorem~\ref{thm:causal_location} does not exclude a larger controlled process containing a past classical register; it identifies exactly why that enlarged realization is causally different from conditioning one fixed process in the future.

The result does not imply a probability set beyond quantum theory or the absence of every enlarged quantum realization.  Every setting separately has a standard quantum realization, and Theorem~\ref{thm:causal_location} constructs one larger past-controlled realization of the complete family. The obstruction concerns future conditioning of one deterministic parent on the specified $A,B$ laboratory slots.

The minimal definite-order example admits a single ordinary circuit with a past control and a fixed future record measurement.  Write $x=\mathrm R,\mathrm D$ and define
\begin{equation}
f_{\mathrm R}(j)=0, \qquad f_{\mathrm D}(j)=j, \qquad j\in\{0,1\}.
\end{equation}
The reset and dephasing channels have Stinespring isometries~[S9]
\begin{equation}
V_x|j\rangle^{A_O} = |f_x(j)\rangle^{B_I}|j\rangle^E.
\label{eq:reset_dephase_isometries}
\end{equation}
A single past-controlled isometry is obtained by coherently retaining the classical label:
\begin{equation}
V|x\rangle^X|j\rangle^{A_O} = |x\rangle^{X'}|f_x(j)\rangle^{B_I}|j\rangle^E.
\label{eq:controlled_isometry}
\end{equation}
The output states on the right-hand side are orthonormal for distinct $(x,j)$, hence $V^\dagger V=I$.  Purify the terminal discard of $B_O$ into a qubit $F$.  Conditioned on the past value $x$, the resulting process vector is
\begin{equation}
|\Xi_x\rangle = \sum_{j,\ell=0}^{1} |0,j,f_x(j),\ell\rangle_{A_IA_OB_IB_O} |j,\ell\rangle_{EF}.
\label{eq:explicit_controlled_process_vector}
\end{equation}
Tracing out the future systems gives
\begin{equation}
\operatorname{Tr}_{EF}|\Xi_x\rangle\langle\Xi_x|=W_x.
\label{eq:explicit_parent_recovery}
\end{equation}

The same binary POVM is used for both past settings:
\begin{equation}
E_\pm = |\pm\rangle\langle\pm|^E\otimes|0\rangle\langle0|^F + \frac12 I^E\otimes|1\rangle\langle1|^F,
\label{eq:fixed_future_povm}
\end{equation}
where $|\pm\rangle=(|0\rangle\pm|1\rangle)/\sqrt2$.  It is positive and obeys $E_++E_-=I^{EF}$.  Define
\begin{equation}
|r_{j\ell}^{(x)}\rangle = |0,j,f_x(j),\ell\rangle.
\end{equation}
Since the POVM is real in the displayed basis, the transpose in the Choi contraction is immaterial, and direct evaluation gives
\begin{align}
&\operatorname{Tr}_{EF} \left[ (I\otimes E_\pm)|\Xi_x\rangle\langle\Xi_x| \right] \nonumber\\
&= \frac12 \left( |r_{00}^{(x)}\rangle \pm|r_{10}^{(x)}\rangle \right) \left( \langle r_{00}^{(x)}| \pm\langle r_{10}^{(x)}| \right) \nonumber\\
&\quad+ \frac12\sum_{j=0}^{1} |r_{j1}^{(x)}\rangle\langle r_{j1}^{(x)}| \nonumber\\
&= M_{\pm|x}+\frac12(W_x-R_x) =G_{\pm|x}.
\label{eq:explicit_controlled_branch_recovery}
\end{align}
Here $|r_{00}^{(x)}\rangle=|q\rangle$, while $|r_{10}^{(\mathrm R)}\rangle=|u_{\mathrm R}\rangle$ and $|r_{10}^{(\mathrm D)}\rangle=|u_{\mathrm D}\rangle$.  The actual pair $(x,a)$ can be copied to orthogonal pointer states, yielding a commuting future record algebra; this is not a joint record of outcomes for settings that were not performed.  Thus the complete two-setting family is realized by one past-controlled circuit and one fixed future measurement, but not by choosing a future measurement on one $x$-independent $A,B$ parent.

The block-diagonal construction in Theorem~\ref{thm:causal_location} is exact within the process-matrix formalism.  For causally ordered combs it has an ordinary memory-channel realization~[S7]; Eq.~\eqref{eq:controlled_isometry} supplies one here.  We do not infer an ordinary reversible circuit realization for every abstract indefinite-order process, since not every valid process matrix satisfies the stronger purifiability requirement~[S10].

\section{Finite-data certification}
\label{app:finite_certification}

For exact reconstructed data, Eqs.~(13) and (16) of the main text provide a complete test.  The sparse witness also permits a direct finite-sample test without reconstructing the QSOST operators.  Index the eight quadratures in Eq.~(30) of the main text by $k=1,\ldots,8$, with signs
\begin{equation}
(c_1,\ldots,c_8)=(+1,+1,+1,+1,+1,+1,-1,+1).
\end{equation}
Write $k=(x(k),\ell(k))$, where $x(k)\in\{\mathrm R,\mathrm D\}$ is the record setting and $\ell(k)$ collects the local Pauli pair and the probe quadrature. Here $k$ is a scored data cell, not one control word that must be selected before the laboratories.  In the delayed-choice protocol, $\ell$ is sampled first and fixes the interferometric operations; only after $A,B$ close is $x$ sampled with fresh randomness, after which the trial is assigned to the cell $(x,\ell)$.  The fixed-block test below certifies the witness under its iid sampling assumptions.  A delayed-choice audit with randomized late $x$ can instead use the two-stage importance-weighted test below.

\paragraph{Certification model.}
Here ``causally agnostic'' refers only to the internal process and the causal relation between its laboratory slots.  The present test is causal-structure agnostic but device dependent.  Identifying the measured correlators with Equation~(30) of the main text assumes trusted two-dimensional laboratory interfaces, calibrated Pauli interventions, a calibrated coherent reference arm and phase convention, trusted probe $X/Y$ readout, and correctly identified record settings and outcomes.  The witness coefficients must be fixed before testing.  Calibration margins must either be fixed before testing or generated by a predeclared predictable rule satisfying the uniform coverage assumptions below; they may not be tuned retrospectively to current or future witness outcomes. For each record setting $x$, the terminal record instrument is fixed across all interferometric choices $k$ that use that setting and acts only on its designated record system.  The ideal terminal readout therefore factorizes as
\begin{equation}
M_{s,a|k}^{PR} = M_{s|k}^{P}\otimes E_{a|x(k)}^{R},
\label{eq:probe_record_factorization}
\end{equation}
where $E_{a|x}^{R}$ is independent of $\ell$.  The record apparatus may receive the setting $x$ required to select $E_{a|x}^{R}$, but, before both outputs are fixed, it has no classical access to $\ell(k)$ or to the probe outcome $s$.  No joint adaptive readout or feed-forward between the probe and record measurements is allowed.  The assignment of record labels to scores is fixed for each $x$, independently of $\ell$, before testing.  The ideal inequality does not otherwise require a calibration model for the effects $E_{a|x}^{R}$; they may remain uncharacterized subject to this factorization, completeness, and no-side-channel condition. Setting-dependent postselection is not allowed.  Before data collection, every loss or abort event must be assigned a fixed score in $[-1,1]$, and its induced deviation from the ideal binary correlator must be included in the calibration bounds.  No such event may be removed from the denominator $N_k$.

Before applying the test, validate the trusted conditioned reconstruction, settingwise process realizability, and equality of the unconditional QSOSTs. Then test $g_{\mathrm{fix}}=0$ with full data, or use the eight-quadrature bound for the qubit benchmark.  With finite tomography the confidence region, not just a plug-in estimate, must be disjoint from the common-realizable set. For the late-choice audit described above, randomly interleave probe labels, log the fresh setting $x$ and slot-closure times, and evaluate separate time blocks to distinguish leakage, memory, and drift.

After a calibrated violation, test in order for recalibration or postselection, probe--record leakage, source--setting correlation, and memory or drift.  If these failures are independently bounded beyond the certified margins, then within standard quantum theory the declared-slot model must include setting-dependent control.  This diagnosis does not identify a particular mechanism; its realization-theoretic boundary is given in Appendix~\ref{app:causal_records}.

The sparse inequality is not a device-independent causal inequality.  If the intervention labels and probe-record measurements were untrusted, a response box could output $s_kr_k=c_k$ for each $k$, attaining $\mathcal W=\frac14\sum_kc_k^2=2$ even with a fixed trivial common process. Violation of Eq.~(31) of the main text therefore certifies failure of common-process gluing only relative to the trusted interferometric interfaces and the structural isolation of the record apparatus specified above.  In particular, if the record apparatus could read the full $k$ and $s$, the rule $r=c_ks$ would reproduce the same algebraic value $2$ without testing a common assemblage.

\begin{proposition}[Black-box impossibility and a sufficient trusted boundary]
\label{prop:black_box_impossibility}
In a fully unconstrained device model, the eight observed binary distributions $P(s,r|k)$ alone cannot certify failure of common-process gluing: a joint response box may ignore an arbitrary fixed common process and directly sample any prescribed $P(s,r|k)$.  In particular, it may output any $s\in\{-1,+1\}$ and set
\begin{equation}
r=c_ks.
\label{eq:black_box_spoof}
\end{equation}
It attains the algebraic value
\begin{equation}
\mathcal W = \frac14\sum_{k=1}^{8}c_k\mathbb E[sr|k] =2.
\end{equation}
Consequently no nontrivial threshold below the algebraic maximum has fully device-independent soundness.  A sufficient, but not claimed necessary, partially characterized model does allow the record effects to remain uncharacterized: the probe interfaces and score convention are trusted, while Eq.~\eqref{eq:probe_record_factorization}, $\ell$-independence of $E_{a|x}^{R}$, and the stated no-side-channel condition are enforced.
\end{proposition}

\begin{proof}
An unconstrained joint response box can sample the specified conditional distribution directly, independently of the underlying process. Equation~\eqref{eq:black_box_spoof} gives $\mathbb E[sr|k]=c_k$ for every $k$, independently of the process, and hence the displayed value.  This proves the impossibility statement.  Under the factorization and no-side-channel assumptions, the record effect for a fixed $x$ is the same effect in every quadrature in which it appears and cannot condition its response on $\ell$ or $s$.  The joint differences therefore have the branch-interference form of Proposition~\ref{prop:branch_process_readout}; no numerical characterization of the individual effects is used in the common-process witness bound.
\end{proof}

In trial $n$ of setting $k$, encode the probe outcome by $s_{k,n}\in\{+1,-1\}$ and the record outcome by $r_{k,n}\in\{+1,-1\}$.  Equations~\eqref{eq:real_amplitude} and \eqref{eq:imag_amplitude} give
\begin{equation}
\Delta I_k=\mathbb E[s_kr_k].
\end{equation}
In the ideal lossless experiment set $z_{k,n}=s_{k,n}r_{k,n}$; more generally, let $z_{k,n}\in[-1,1]$ denote the preassigned score, including any loss or abort event.  For $N_k$ trials, define
\begin{equation}
\widehat{\mathcal W} = \frac14\sum_{k=1}^{8} \frac{c_k}{N_k} \sum_{n=1}^{N_k}z_{k,n}.
\label{eq:finite_estimator}
\end{equation}

\begin{theorem}[Finite-sample witness test]
\label{thm:finite_sample}
Fix the witness independently of the testing data.  Assume that all testing trials are mutually independent, that the trials within each block $k$ are identically distributed, and that the sample sizes $N_k$ are fixed in advance or independently of the observed outcomes.  Assume also that the eight stationary mean scores are the branch-resolved correlators of one fixed common-process assemblage.  Then, for every $\delta>0$,
\begin{equation}
\Pr\!\left[ \widehat{\mathcal W}\ge\sqrt2+\delta \right] \le \exp\!\left[ -\frac{8\delta^2}{\sum_{k=1}^{8}N_k^{-1}} \right].
\label{eq:hoeffding_bound}
\end{equation}
Consequently, a one-sided level-$\alpha$ rejection rule is
\begin{equation}
\widehat{\mathcal W} > \sqrt2+ \sqrt{ \frac{\log(1/\alpha)}{8} \sum_{k=1}^{8}\frac1{N_k} }.
\label{eq:finite_rejection}
\end{equation}
For equal sample sizes $N_k=N$, the statistical margin is $\sqrt{\log(1/\alpha)/N}$.
\end{theorem}

\begin{proof}
Under the stated null model the eight means define one common-process assemblage, so their witness expectation obeys $\mathcal W\le\sqrt2$.  Each single-trial contribution $c_kz_{k,n}/(4N_k)$ has range length at most $1/(2N_k)$. Hoeffding's inequality applied to the sum gives Eq.~\eqref{eq:hoeffding_bound}; solving for $\delta$ gives Eq.~\eqref{eq:finite_rejection}.
\end{proof}

\begin{theorem}[Calibration-robust finite-sample rejection]
\label{thm:calibration_robust}
Let
\begin{equation}
m_k=\mathbb E_{P_k}[z_k], \qquad \widetilde m_k=\mathbb E_{\widetilde P_k}[z_k],
\end{equation}
where $P_k$ is the ideal joint probe-record distribution and $\widetilde P_k$ is the implemented testing distribution for quadrature $k$. Suppose that the ideal distributions are generated by one common-process assemblage, so
\begin{equation}
\frac14\sum_{k=1}^{8}c_km_k\le\sqrt2.
\end{equation}
Let an independent pretest calibration procedure determine $\varepsilon_1,\ldots,\varepsilon_8$.  Assume that, under the common-process null hypothesis, it has the simultaneous test-time coverage property
\begin{equation}
\Pr_{\mathrm{cal}} \left[ |\widetilde m_k-m_k|\le\varepsilon_k \ \text{for all }k=1,\ldots,8 \right] \ge1-\beta.
\label{eq:test_time_calibration_coverage}
\end{equation}
Here $\widetilde m_k$ is the mean of the distribution operating during the testing run, so $\varepsilon_k$ must include any certified calibration-to-test drift.  Define
\begin{equation}
B_{\mathrm{cal}} = \frac14\sum_{k=1}^{8}\varepsilon_k.
\label{eq:calibration_margin}
\end{equation}
Conditional on the calibration data and on the implemented distributions, assume only the sampling conditions used in Theorem~\ref{thm:finite_sample}: all testing trials are mutually independent, the trials within block $k$ are identically distributed according to $\widetilde P_k$, and the sample sizes $N_k$ are fixed in advance or are independent of the observed outcomes.  The implemented distributions $\widetilde P_k$ themselves are not assumed to form a common-process assemblage.  Then, under the common-process null,
\begin{equation}
\Pr\!\left[ \widehat{\mathcal W} \ge \sqrt2+B_{\mathrm{cal}}+\delta \right] \le \beta+ \exp\!\left[ -\frac{8\delta^2}{\sum_{k=1}^{8}N_k^{-1}} \right].
\label{eq:calibration_robust_tail}
\end{equation}
Consequently, if $\alpha_{\mathrm{test}}+\beta\le\alpha$, the rule
\begin{equation}
\widehat{\mathcal W} > \sqrt2+B_{\mathrm{cal}} + \sqrt{ \frac{\log(1/\alpha_{\mathrm{test}})}{8} \sum_{k=1}^{8}\frac1{N_k} }
\label{eq:calibration_robust_rejection}
\end{equation}
has type-I error at most $\alpha$.
\end{theorem}

\begin{proof}
Condition on the calibration data.  On the event in Eq.~\eqref{eq:test_time_calibration_coverage},
\begin{align}
\mathbb E[\widehat{\mathcal W}] &= \frac14\sum_{k=1}^{8}c_k\widetilde m_k \nonumber\\
&\le \frac14\sum_{k=1}^{8}c_km_k + \frac14\sum_{k=1}^{8}|\widetilde m_k-m_k| \nonumber\\
&\le \sqrt2+B_{\mathrm{cal}}.
\end{align}
The same range calculation as in Theorem~\ref{thm:finite_sample} bounds the conditional testing tail by the exponential term in Eq.~\eqref{eq:calibration_robust_tail}.  Adding the calibration-failure probability $\beta$ proves the result.
\end{proof}

With the convention
\begin{equation}
\operatorname{TV}(P,Q) = \frac12\sum_{\omega}|P(\omega)-Q(\omega)|,
\end{equation}
where $\omega$ ranges over all scored outcomes, a simultaneous calibration guarantee $\operatorname{TV}(\widetilde P_k,P_k)\le\tau_k$ implies $|\widetilde m_k-m_k|\le2\tau_k$ and hence
\begin{equation}
B_{\mathrm{cal}} = \frac12\sum_{k=1}^{8}\tau_k.
\end{equation}
These bounds must hold uniformly for the tested interfaces; agreement on a small set of calibration states alone does not establish such a guarantee. Deterministic error bounds correspond to $\beta=0$.  If the eight separate calibration guarantees have failure probabilities $\beta_k$, the union bound gives a simultaneous failure probability at most $\min\{1,\sum_k\beta_k\}$; one may therefore preallocate $\sum_k\beta_k\le\beta$.

\subsection{Anytime-valid certification with memory and adaptive settings}

The iid assumption can be removed if the common-process null is imposed conditionally before every late record-setting choice.  Let $\mathcal F_{t-1}$ contain the calibration data and the complete observed history before trial $t$. The experimenter chooses a predictable distribution
\begin{equation}
\pi_{t,k} = \Pr(K_t=k|\mathcal F_{t-1}), \qquad k=1,\ldots,8,
\label{eq:adaptive_setting_distribution}
\end{equation}
with
\begin{equation}
q_t := \min_k\pi_{t,k} >0.
\label{eq:positive_exploration}
\end{equation}
For a delayed-choice implementation, the components of $K_t=(x_t,\ell_t)$ are generated in two stages: $\ell_t$ is drawn before the laboratories, whereas $x_t$ is drawn with independent fresh randomness only after both laboratories close.  Their predictable joint probability is $\pi_{t,k}$, and the trial-$t$ process-side strategy must be fixed before the late draw.  The randomization probabilities may otherwise depend arbitrarily on $\mathcal F_{t-1}$.

Let $Z_t\in[-1,1]$ be the predeclared probe-record score in trial $t$. The implemented conditional means are
\begin{equation}
\widetilde m_{t,k} = \mathbb E[Z_t|\mathcal F_{t-1},K_t=k].
\label{eq:implemented_conditional_means}
\end{equation}
The sequential common-process null requires that, for every $t$ and every history of nonzero probability, there exist ideal counterfactual means $m_{t,1},\ldots,m_{t,8}$ generated by one common-process assemblage conditional on that history, such that
\begin{equation}
\frac14\sum_{k=1}^{8}c_km_{t,k} \le \sqrt2.
\label{eq:conditional_common_process_null}
\end{equation}
The conditional assemblage and all eight means may depend arbitrarily on $\mathcal F_{t-1}$; no independence, stationarity, or absence of memory across trials is assumed.  Here ``counterfactual means'' denotes the eight conditional response distributions defined by the randomized-setting protocol.  It does not assert simultaneous pre-existing outcomes or the classical refinement excluded by Proposition~\ref{prop:no_counterfactual_record}.

Let a calibration procedure provide nonnegative predictable numbers $\varepsilon_{t,k}$.  Assume that, for every $t,k$ and every allowed history, an $\mathcal F_0$-measurable calibration-valid event $\mathcal C_{\mathrm{cal}}$ satisfies
\begin{align}
c_k(\widetilde m_{t,k}-m_{t,k}) &\le \varepsilon_{t,k},
\label{eq:sequential_calibration_bound}\\
\Pr(\mathcal C_{\mathrm{cal}}) &\ge 1-\beta.
\label{eq:sequential_calibration_coverage}
\end{align}
The event $\mathcal C_{\mathrm{cal}}$ is determined entirely before the testing sequence and is not selected using future testing outcomes.  For each calibration realization in this event, the bias inequalities above are assumed to hold as deterministic uniform statements over all subsequent allowed histories. Absolute bias bounds are sufficient.  Define
\begin{equation}
B_t = \frac14\sum_{k=1}^{8}\varepsilon_{t,k}, \qquad b_t = \sqrt2+B_t.
\label{eq:sequential_calibration_margin}
\end{equation}

\begin{theorem}[Importance-weighted anytime-valid witness test]
\label{thm:anytime_witness}
Under Eqs.~\eqref{eq:adaptive_setting_distribution}-- \eqref{eq:sequential_calibration_margin}, define
\begin{equation}
X_t = \frac{c_{K_t}Z_t}{4\pi_{t,K_t}}.
\label{eq:importance_weighted_score}
\end{equation}
Let $\lambda_t\ge0$ be any finite $\mathcal F_{t-1}$-measurable betting coefficient chosen before $K_t$ and $Z_t$ are observed.  For every calibration realization in $\mathcal C_{\mathrm{cal}}$,
\begin{equation}
E_t = \exp \left\{ \sum_{i=1}^{t} \left[ \lambda_i(X_i-b_i) - \frac{\lambda_i^2}{32q_i^2} \right] \right\}, \qquad E_0=1,
\label{eq:anytime_e_process}
\end{equation}
is a nonnegative supermartingale under the conditional common-process null. Consequently,
\begin{align}
\Pr\left[ \mathcal C_{\mathrm{cal}} \ \text{and}\ \sup_{t\ge1}E_t \ge \frac1{\alpha_{\mathrm{seq}}} \right] &\le \alpha_{\mathrm{seq}},
\label{eq:ville_bound}\\
\Pr\left[ \sup_{t\ge1}E_t \ge \frac1{\alpha_{\mathrm{seq}}} \right] &\le \alpha_{\mathrm{seq}}+\beta.
\label{eq:calibration_adjusted_ville}
\end{align}
Thus rejection at the first time $E_t\ge1/\alpha_{\mathrm{seq}}$ is valid under arbitrary optional stopping whenever $\alpha_{\mathrm{seq}}+\beta\le\alpha$.
\end{theorem}

\begin{proof}
Fresh randomized setting selection and importance weighting give
\begin{align}
\mathbb E[X_t|\mathcal F_{t-1}] &= \sum_{k=1}^{8} \pi_{t,k} \frac{c_k}{4\pi_{t,k}} \widetilde m_{t,k} \nonumber\\
&= \frac14\sum_{k=1}^{8}c_k\widetilde m_{t,k} \le b_t
\label{eq:conditional_score_bound}
\end{align}
on $\mathcal C_{\mathrm{cal}}$.  Moreover,
\begin{equation}
-\frac1{4q_t} \le X_t \le \frac1{4q_t},
\end{equation}
so its conditional range length is at most $1/(2q_t)$.  Conditional Hoeffding's lemma yields
\begin{equation}
\mathbb E \left[ \exp \left\{ \lambda_t \left( X_t-\mathbb E[X_t|\mathcal F_{t-1}] \right) \right\} \middle| \mathcal F_{t-1} \right] \le \exp \left( \frac{\lambda_t^2}{32q_t^2} \right).
\end{equation}
Together with Eq.~\eqref{eq:conditional_score_bound}, this shows that every multiplicative factor in Eq.~\eqref{eq:anytime_e_process} has conditional expectation at most one.  Ville's inequality proves Eq.~\eqref{eq:ville_bound}; adding $\Pr(\mathcal C_{\mathrm{cal}}^c)\le\beta$ proves Eq.~\eqref{eq:calibration_adjusted_ville}.
\end{proof}

For a fixed predeclared $\lambda>0$, set
\begin{equation}
S_t = \sum_{i=1}^{t}(X_i-b_i), \qquad V_t = \sum_{i=1}^{t}q_i^{-2}.
\end{equation}
An anytime-valid rejection occurs whenever
\begin{equation}
S_t \ge \frac{\log(1/\alpha_{\mathrm{seq}})}{\lambda} + \frac{\lambda V_t}{32}.
\label{eq:fixed_lambda_anytime_boundary}
\end{equation}
The betting coefficient cannot be optimized after observing the same data. A horizon-free alternative mixes the fixed-$\lambda$ e-processes.  For $\rho>0$, the half-normal density
\begin{equation}
g_\rho(\lambda) = \sqrt{\frac{2\rho}{\pi}} \exp \left( -\frac{\rho\lambda^2}{2} \right), \qquad \lambda\ge0,
\end{equation}
gives
\begin{align}
E_t^{\mathrm{mix}} &= \int_0^\infty g_\rho(\lambda) \exp \left( \lambda S_t-\frac{\lambda^2V_t}{32} \right) d\lambda \nonumber\\
&= \sqrt{ \frac{\rho}{\rho+V_t/16} } \exp \left[ \frac{S_t^2}{2(\rho+V_t/16)} \right] \nonumber\\
&\quad\times \left[ 1+ \operatorname{erf} \left( \frac{S_t}{\sqrt{2(\rho+V_t/16)}} \right) \right].
\label{eq:mixture_e_process}
\end{align}
This mixture starts at one and is again an e-process.  Hence
\begin{equation}
p_t^{\mathrm{seq}} = \min \left\{ 1,\ \left( \max_{0\le u\le t}E_u^{\mathrm{mix}} \right)^{-1} \right\}
\label{eq:anytime_p_value}
\end{equation}
is anytime valid on $\mathcal C_{\mathrm{cal}}$, while
\begin{equation}
p_t^{\mathrm{total}} = \min \left\{ 1,\ p_t^{\mathrm{seq}}+\beta \right\}
\end{equation}
is valid after including calibration failure.

If, simultaneously for every allowed history and trial,
\begin{equation}
\operatorname{TV} \left( \widetilde P_{t,k}(\,\cdot\,|\mathcal F_{t-1}), P_{t,k}(\,\cdot\,|\mathcal F_{t-1}) \right) \le \tau_{t,k},
\end{equation}
then one may take
\begin{equation}
\varepsilon_{t,k}=2\tau_{t,k}, \qquad B_t=\frac12\sum_{k=1}^{8}\tau_{t,k}.
\end{equation}
For a finite horizon, pretest calibration failures may be allocated so that $\sum_{t,k}\beta_{t,k}\le\beta$, provided each component guarantee controls the corresponding test-time bias uniformly over the histories to which it is applied and all calibration data are included in $\mathcal F_0$.  An unbounded horizon under the theorem as stated requires an $\mathcal F_0$-measurable deterministic or simultaneous pretest device bound.  An online calibration confidence sequence depends on future data and is not covered merely by conditioning on $\mathcal C_{\mathrm{cal}}$; it must instead be combined with the witness through an independent or joint e-process proof. Fresh setting randomness, $q_t>0$, and time-uniform calibration cannot be removed: without them an unmeasured witness term, a setting-dependent parent, or an arbitrary test-time readout drift can reproduce a violation even for a fixed trivial common process.

The equal-dimension restriction has a partial, reference-dependent resolution.  If every local output dimension is at least its input dimension, fixed local reference isometries define a partial-swap process compression whose normalization, column identity, least positive lift, and compression--gluing theorem retain the same form.  This object is not the canonical QSOST of Ref.~[S11].  If any local output dimension is smaller than its input dimension, no rank-one trace-preserving reference exists.  Appendix~\ref{app:rectangular} gives the construction and the precise boundary.

\section{Channel specialization of the conditioned probe state}
\label{app:probe}

Proposition~\ref{prop:branch_process_readout} establishes the branch-resolved readout for a general bipartite process.  For completeness, this appendix gives its familiar causally ordered channel specialization.  Let the probe start in $|+\rangle_P$ and let the two arms implement the reference and intervention branches.  A Kraus representation of branch $a|x$ is $\{K_{a\lambda|x}\}_\lambda$.  The off-diagonal probe block contains the complex amplitude
\begin{equation}
I_{\mu\nu,a|x} = \sum_\lambda \operatorname{Tr} \left[ W_\nu^B K_{a\lambda|x}U_\mu^A\rho_AK_{a\lambda|x}^\dagger \right].
\end{equation}
The non-normalized probe state is
\begin{equation}
\tau_{P,a|x}^{\mu\nu} = \frac12
\begin{pmatrix}
p_{0,a|x}&I_{\mu\nu,a|x}^{*}\\
I_{\mu\nu,a|x}&p_{1,a|x}^{\mu}
\end{pmatrix},
\end{equation}
where
\begin{align}
p_{0,a|x}&=\operatorname{Tr}[\mathcal E_{a|x}(\rho_A)],\\
p_{1,a|x}^{\mu}&=\operatorname{Tr}[\mathcal E_{a|x}(U_\mu^A\rho_A(U_\mu^A)^\dagger)].
\end{align}
Probe measurements in the $X$ and $Y$ bases give Eqs.~\eqref{eq:real_amplitude} and \eqref{eq:imag_amplitude}.  Positivity implies
\begin{equation}
|I_{\mu\nu,a|x}|^2 \le p_{0,a|x}p_{1,a|x}^{\mu}.
\end{equation}
This is a branch-level interferometric visibility constraint and should not be confused with the common-process gluing witness in Eq.~(31) of the main text.

\section{Process projector and QSOST adjoint}
\label{app:process_conventions}

For a subsystem $X$, the trace-and-replace map is
\begin{equation}
{}_X W = \frac{I_X}{d_X}\otimes\operatorname{Tr}_X W,
\label{eq:trace_replace}
\end{equation}
with the factors returned to their original order.  The orthogonal valid-process projector is
\begin{align}
L_V(W)={}&{}_{{A_O}}W+{}_{{B_O}}W-{}_{{A_OB_O}}W \nonumber\\
&-{}_{{A_IA_O}}W+{}_{{A_IA_OB_O}}W \nonumber\\
&-{}_{{B_IB_O}}W+{}_{{A_OB_IB_O}}W.
\label{eq:LV}
\end{align}
Direct application to the reference vector gives
\begin{equation}
L_V(|\Omega\rangle\langle\Omega|)=\frac{I}{d_O}.
\label{eq:LV_Omega}
\end{equation}

The standard bipartite process constraints can be written as
\begin{align}
[1-B_O]A_IA_OW&=0,\\
[1-A_O]B_IB_OW&=0,\\
[1-A_O][1-B_O]W&=0.
\end{align}
The trace-and-replace maps commute, are self-adjoint, and are idempotent.  Multiplying the three corresponding orthogonal projectors and expanding gives Eq.~\eqref{eq:LV}.

In the swap-Jamio\l kowski convention of Ref.~[S11], the process operator is related to the standard positive matrix by partial transpose on the grouped input spaces.  Substitution into the first-order QSOST formula gives Eq.~(3) of the main text. In grouped matrix elements the resulting contraction is
\begin{equation}
[Q_+(G)]_{\alpha\beta} = \sum_i\overline G_{ii,\alpha\beta},
\end{equation}
which fixes both the partial-transpose basis and the input--output ordering used throughout. Under the fixed identification of the isomorphic grouped spaces, it also gives, for arbitrary operators $X_I$ and $Y_O$,
\begin{equation}
Q_+(X_I\otimes Y_O)=X_I^T Y_O.
\label{eq:Q_product}
\end{equation}

The complex Hilbert--Schmidt adjoint and the effective process-space coefficient for a real Hermitian witness pairing are, respectively,
\begin{equation}
Q_+^\dagger(F) = \Pi^\dagger \left[ S_{I:O}(I_I\otimes F) \right]^{T_I} \Pi,
\label{eq:Q_adjoint}
\end{equation}
and
\begin{equation}
L_V\!\left( \frac{Q_+^\dagger(F)+Q_+^\dagger(F)^\dagger}{2} \right).
\label{eq:real_adjoint}
\end{equation}

For the adjoint, use the Hilbert--Schmidt pairing:
\begin{align}
\operatorname{Tr}[F^\dagger Q_+(W)] &= \operatorname{Tr}[(I_I\otimes F^\dagger)S_{I:O}\overline W^{T_I}]\\
&= \operatorname{Tr} \left( \left[(I_I\otimes F^\dagger)S_{I:O}\right]^{T_I} \overline W \right).
\end{align}
Taking the adjoint of the coefficient and restoring the interleaved order yields Eq.~\eqref{eq:Q_adjoint}.

Finally, applying Eq.~\eqref{eq:LV} to $|\Omega\rangle\langle\Omega|$ gives Eq.~\eqref{eq:LV_Omega}.  Self-adjointness of $L_V$ then gives, for $L_V(W)=W$,
\begin{equation}
\langle\Omega|W|\Omega\rangle = \operatorname{Tr}[L_V(|\Omega\rangle\langle\Omega|)W] = \frac{\operatorname{Tr}W}{d_O}.
\end{equation}

\section{Operator-basis proof of the hidden-fiber dimension}
\label{app:kernel}

Let $\operatorname{Herm}_d$ be the real vector space of Hermitian operators on a $d$-dimensional space and $\operatorname{Herm}_d^0$ its traceless subspace.  Choose on each process factor an orthogonal Hermitian operator basis whose zeroth element is the identity and whose remaining elements span $\operatorname{Herm}_d^0$.  For a tensor-basis element, write $0$ when a factor is the identity and $1$ when it is traceless.  Since a trace-and-replace map retains the former and annihilates the latter, direct application of Eq.~\eqref{eq:LV} shows that $\operatorname{ran}L_V\cap\operatorname{Herm}$ contains exactly the support patterns
\begin{equation}
0000, \ 0010, \ 0110, \ 1000, \ 1001, \ 1010, \ 1011, \ 1110.
\label{eq:general_valid_patterns}
\end{equation}
With $\alpha=m^2-1$ and $\beta=n^2-1$, their multiplicities are
\begin{equation}
1,\ \beta,\ \alpha\beta,\ \alpha,\ \alpha\beta,\ \alpha\beta,\ \alpha\beta^2,\ \alpha^2\beta.
\end{equation}
Hence
\begin{equation}
D_V = 1+\alpha+\beta+3\alpha\beta+\alpha^2\beta+\alpha\beta^2.
\label{eq:general_valid_dimension}
\end{equation}

Equation~\eqref{eq:Q_product} gives
\begin{equation}
Q_+ (X_{A_I}\otimes Y_{A_O}\otimes Z_{B_I}\otimes T_{B_O}) = (X^TY) \otimes (Z^TT).
\label{eq:general_basis_image}
\end{equation}
The patterns $0000,0010,1000,1010$ already map onto every Hermitian operator on $A_OB_O$: their images supply, respectively, identity--identity, identity--traceless, traceless--identity, and traceless--traceless tensor-basis elements.

It remains to determine which anti-Hermitian directions can occur.  For Hermitian $Z,T$,
\begin{equation}
Z^TT = \frac12\{Z^T,T\}+\frac12[Z^T,T],
\end{equation}
where the commutator is anti-Hermitian and traceless.  Among the valid patterns, only $1011$ and $1110$ contain a nontrivial input--output pair within one laboratory.  Equation~\eqref{eq:general_basis_image} therefore shows that every anti-Hermitian image belongs to
\begin{equation}
i\left( \operatorname{Herm}_m^0\otimes\operatorname{Herm}_n^0 \right).
\label{eq:antihermitian_image_containment}
\end{equation}
Conversely, this commutator span can be checked explicitly.  With matrix units $E_{jk}$, define for $j<k$
\begin{align}
X_{jk}&=E_{jk}+E_{kj},\nonumber\\
Y_{jk}&=-i(E_{jk}-E_{kj}),\nonumber\\
Z_{jk}&=E_{jj}-E_{kk}.
\end{align}
These operators are traceless Hermitian, and
\begin{align}
[X_{jk},Y_{jk}]&=2iZ_{jk},\nonumber\\
[Z_{jk},X_{jk}]&=2iY_{jk},\nonumber\\
[Y_{jk},Z_{jk}]&=2iX_{jk}.
\label{eq:explicit_commutator_basis}
\end{align}
The $X_{jk},Y_{jk}$ together with $Z_{j,j+1}$ form a real basis of $\operatorname{Herm}_d^0$. Equation~\eqref{eq:explicit_commutator_basis} therefore proves, without an appeal to the abstract commutator-algebra identity, that commutators of traceless Hermitian matrices span $i\operatorname{Herm}_d^0$.

To apply this to the valid process basis, let $F\in\operatorname{Herm}_m^0$ and choose traceless Hermitian $Z,T$ on $B$ so that $[Z^T,T]$ is any desired basis element of $i\operatorname{Herm}_n^0$; transposition preserves the traceless Hermitian space, so Eq.~\eqref{eq:explicit_commutator_basis} supplies all such choices. The valid support pattern $1011$, represented by $F\otimes I\otimes Z\otimes T$, has image
\begin{equation}
F^T\otimes Z^TT = \frac12F^T\otimes\{Z^T,T\} + \frac12F^T\otimes[Z^T,T].
\end{equation}
As $F^T$ spans $\operatorname{Herm}_m^0$, the second term spans every direction in Eq.~\eqref{eq:antihermitian_image_containment}.  The first term is Hermitian and lies in the already established full Hermitian image, so it can be subtracted using the patterns $0000,0010,1000,1010$.  This leaves each pure anti-Hermitian basis direction explicitly. If $\alpha\beta=0$, this additional summand is simply zero.  Hence, as a real vector space,
\begin{equation}
Q_+\!\left( \operatorname{ran}L_V\cap\operatorname{Herm} \right) = \operatorname{Herm}_{mn} \mathbin{\oplus} i\left( \operatorname{Herm}_m^0\otimes\operatorname{Herm}_n^0 \right).
\label{eq:general_real_image}
\end{equation}
The image dimension is therefore
\begin{equation}
D_Q=m^2n^2+\alpha\beta.
\label{eq:general_image_dimension}
\end{equation}
Rank--nullity now gives
\begin{align}
D_V-D_Q &= \alpha\beta(1+\alpha+\beta) \nonumber\\
&= (m^2-1)(n^2-1)(m^2+n^2-1),
\end{align}
which proves Proposition~3 of the main text.

For an explicit hidden direction, fix two orthonormal basis vectors $|0\rangle,|1\rangle$ in every nontrivial local space and set
\begin{equation}
S_d=|0\rangle\langle1|+|1\rangle\langle0|, \qquad P_d=S_d^2.
\end{equation}
Define
\begin{align}
\Delta_{m,n} ={}& S_m^{A_I}\otimes I^{A_O}\otimes \bigl( S_n^{B_I}\otimes S_n^{B_O} -P_n^{B_I}\otimes I^{B_O} \bigr).
\label{eq:general_hidden_direction}
\end{align}
Then the normalized same-fiber pair
\begin{equation}
W_\pm=\frac{I}{mn}\pm\epsilon\Delta_{m,n}, \qquad 0<\epsilon<\frac{1}{2mn},
\label{eq:general_hidden_pair}
\end{equation}
consists of distinct positive-definite deterministic processes of order $B\prec A$.

For completeness, we verify the explicit direction in Eq.~\eqref{eq:general_hidden_direction}.  Its first term has support type $1011$.  Since
\begin{equation}
P_n=\frac2n I_n+\left(P_n-\frac2n I_n\right),
\end{equation}
the second term is a sum of support types $1000$ and $1010$.  All three types belong to $\operatorname{ran}L_V$ and are identity on $A_O$. Therefore
\begin{equation}
L_V(\Delta_{m,n})=\Delta_{m,n}, \qquad {}_{A_O}\Delta_{m,n}=\Delta_{m,n},
\end{equation}
so the direction belongs to the $B\prec A$ linear subspace.  It is manifestly Hermitian and is traceless because $\operatorname{Tr}S_m=0$. Equation~\eqref{eq:Q_product} gives
\begin{equation}
Q_+(\Delta_{m,n}) = S_m^T\otimes(S_n^TS_n-P_n^T) =0.
\end{equation}
If $s_i\in\{1,-1,0\}$ are the eigenvalues of $S_n$, the eigenvalues of $S_n\otimes S_n-P_n\otimes I_n$ are $s_is_j-s_i^2$.  Their largest absolute value is two, so $\|\Delta_{m,n}\|_\infty=2$.  Finally, $\operatorname{Tr}[I/(mn)]=mn=d_O$, and
\begin{equation}
\lambda_{\min}(W_\pm) \ge \frac1{mn}-2\epsilon>0.
\end{equation}
This proves every assertion following Eq.~\eqref{eq:general_hidden_pair}.

\section{Reference-isometry compression beyond equal dimensions}
\label{app:rectangular}

This appendix separates the part of the construction that extends to rectangular local slots from the part that is specific to the canonical equal-dimension QSOST.  Suppose
\begin{equation}
d_{A_O}\ge d_{A_I}, \qquad d_{B_O}\ge d_{B_I},
\label{eq:rectangular_dimension_condition}
\end{equation}
and fix local isometries
\begin{equation}
J_A:A_I\longrightarrow A_O, \qquad J_B:B_I\longrightarrow B_O, \qquad J_X^\dagger J_X=I_{X_I}.
\end{equation}
Write $d_I=d_{A_I}d_{B_I}$. In grouped order let
\begin{equation}
J=J_A\otimes J_B:I\longrightarrow O, \qquad P_J=JJ^\dagger.
\end{equation}
For a reference copy $\widetilde I$ of $I$, define
\begin{align}
S_J &= (I_I\otimes J) S_{I:\widetilde I} (I_I\otimes J^\dagger),
\label{eq:partial_swap}\\
Q_J(G) &= \operatorname{Tr}_I \left[ S_J\overline G^{T_I} \right].
\label{eq:rectangular_compression}
\end{align}
The partial swap is Hermitian and obeys
\begin{equation}
S_J^2=I_I\otimes P_J.
\end{equation}
It reduces to the compression in Eq.~(3) of the main text when $J$ is the fixed unitary input--output identification.

Introduce the grouped reference vector and its interleaved representative,
\begin{equation}
|\Omega_J\rangle = \sum_{i=1}^{d_I}|i\rangle_I\otimes J|i\rangle_O, \qquad |\widetilde\Omega_J\rangle = \Pi^\dagger|\Omega_J\rangle.
\label{eq:rectangular_reference_vector}
\end{equation}
For an operator $\Sigma$ on $O$, set
\begin{align}
|\mathfrak v_O(\Sigma)\rangle &= \sum_{\alpha,\beta} \Sigma_{\alpha\beta}^{*} |\alpha\rangle_O\otimes|\beta\rangle_O, \\
|q_J(\Sigma)\rangle &= \Pi^\dagger (J^T\otimes I_O) |\mathfrak v_O(\Sigma)\rangle.
\label{eq:rectangular_q_map}
\end{align}
The transpose in Eq.~\eqref{eq:rectangular_q_map} is essential for a complex isometry and refers to the same fixed bases as the partial transpose in Eq.~\eqref{eq:rectangular_compression}.

\begin{proposition}[Rectangular column identity]
\label{prop:rectangular_column}
For every Hermitian $G$,
\begin{align}
P_JQ_J(G)&=Q_J(G),
\label{eq:rectangular_left_support}\\
|q_J(Q_J(G))\rangle &= G|\widetilde\Omega_J\rangle,
\label{eq:rectangular_column_identity}\\
\operatorname{Tr}Q_J(G) &= \langle\widetilde\Omega_J|G|\widetilde\Omega_J\rangle.
\label{eq:rectangular_trace_identity}
\end{align}
For product operators in grouped order,
\begin{equation}
Q_J(X_I\otimes Y_O) = JX_I^TJ^\dagger Y_O.
\label{eq:rectangular_product_identity}
\end{equation}
If $W$ is a deterministic process, then
\begin{equation}
\langle\widetilde\Omega_J|W|\widetilde\Omega_J\rangle=1, \qquad \operatorname{Tr}Q_J(W)=1.
\label{eq:rectangular_normalization}
\end{equation}
\end{proposition}

\begin{proof}
In grouped bases, with $J_{\alpha i}=\langle\alpha|J|i\rangle$, direct evaluation of Eq.~\eqref{eq:rectangular_compression} gives
\begin{equation}
[Q_J(G)]_{\alpha\beta} = \sum_{r,k,\gamma} J_{\alpha k}J_{\gamma r}^{*} \overline G_{r\gamma,k\beta}.
\label{eq:rectangular_index_formula}
\end{equation}
Multiplication on the left by $P_J$ and $J^\dagger J=I_I$ prove Eq.~\eqref{eq:rectangular_left_support}.  Hermiticity of $G$ gives
\begin{align}
& \left[ (J^T\otimes I_O) \mathfrak v_O(Q_J(G)) \right]_{i\beta} \nonumber\\
&\quad= \sum_{r,\gamma} J_{\gamma r}\overline G_{i\beta,r\gamma} = [\overline G|\Omega_J\rangle]_{i\beta},
\end{align}
which proves Eq.~\eqref{eq:rectangular_column_identity}.  Taking the trace in Eq.~\eqref{eq:rectangular_index_formula} proves Eq.~\eqref{eq:rectangular_trace_identity}.  Substitution of $\overline G=X_I\otimes Y_O$ proves Eq.~\eqref{eq:rectangular_product_identity}.  Finally, $|\widetilde\Omega_J\rangle\langle\widetilde\Omega_J|$ is the product Choi operator of the two local trace-preserving isometry channels.  Process normalization against these deterministic local operations proves Eq.~\eqref{eq:rectangular_normalization}.
\end{proof}

For a physical branch $G\succeq0$ belonging to a positive decomposition of a deterministic process, the reference-dependent compression remains operationally tomographic once the reference isometries have been fixed.  Let $L=L_A\otimes L_B:I\to O$ be any product of local isometries.  The two probe-readout instruments with Kraus operators
\begin{equation}
K_{b,X}=\frac{J+bL}{2}, \qquad K_{b,Y}=\frac{J-ibL}{2}, \qquad b\in\{-1,+1\},
\label{eq:rectangular_probe_kraus}
\end{equation}
are trace preserving after summing over $b$, since
\begin{equation}
\sum_bK_{b,\xi}^\dagger K_{b,\xi}=I_I.
\end{equation}
Their joint outcome differences isolate the real and imaginary parts of
\begin{align}
\mathcal I_L(G) &= \langle\Omega_J|\overline G|\Omega_L\rangle \nonumber\\
&= \operatorname{Tr} \left[ (LJ^\dagger)Q_J(G) \right],
\label{eq:rectangular_interference}
\end{align}
where $|\Omega_L\rangle=\sum_i|i\rangle_I\otimes L|i\rangle_O$. These functionals determine every left-supported image in Eq.~\eqref{eq:rectangular_left_support}.  To see this, first note that local isometries linearly span $\mathcal L(X_I,X_O)$ when $d_{X_O}\ge d_{X_I}$.  Indeed, if a nonzero operator $A$ were Hilbert--Schmidt orthogonal to every isometry, extend the partial isometry in the polar decomposition $A=V|A|$ to an isometry $K$. Then
\begin{equation}
\operatorname{Tr}(K^\dagger A)=\operatorname{Tr}|A|>0,
\end{equation}
a contradiction.  Products of the two local spanning families span $\mathcal L(I,O)$, and Eq.~\eqref{eq:rectangular_interference} then separates all operators satisfying $P_J\Sigma=\Sigma$.

\begin{theorem}[Least lift and gluing for a fixed reference isometry]
\label{thm:rectangular_gluing}
Let $\Sigma$ obey $P_J\Sigma=\Sigma$, and define
\begin{equation}
p=\operatorname{Tr}\Sigma, \qquad |q\rangle=|q_J(\Sigma)\rangle.
\end{equation}
There exists $G\succeq0$ with $Q_J(G)=\Sigma$ if and only if either $p>0$ is real or $\Sigma=0$.  For $p>0$ the unique least positive lift is
\begin{equation}
M_J(\Sigma) = \frac{|q\rangle\langle q|}{p},
\label{eq:rectangular_minimal_lift}
\end{equation}
while $M_J(0)=0$.  Every positive lift has the unique form
\begin{equation}
G=M_J(\Sigma)+H, \qquad H\succeq0, \qquad H|\widetilde\Omega_J\rangle=0.
\label{eq:rectangular_all_lifts}
\end{equation}

Consequently, for any finite family of such left-supported operators with
\begin{equation}
P_J\Sigma_{a|x}=\Sigma_{a|x} \quad\text{for every }a,x,
\end{equation}
$\sum_a\Sigma_{a|x}=T$, $\operatorname{Tr}T=1$, and positive trace for every nonzero branch, define
\begin{equation}
R_x^{J} = \sum_{a:p_{a|x}>0} M_J(\Sigma_{a|x}).
\label{eq:rectangular_Rx}
\end{equation}
Setting $x$ is realizable exactly when some $W_x\in\mathfrak P$ obeys $W_x\succeq R_x^J$, and the whole family has one common parent exactly when one $W\in\mathfrak P$ obeys $W\succeq R_x^J$ for every $x$.  In particular, settingwise realizability implies common-process realizability for every finite conditioned family if and only if $Q_J$ is injective on $\mathfrak P$.
\end{theorem}

\begin{proof}
On the left-supported subspace the map $\Sigma\mapsto q_J(\Sigma)$ is injective: if $q_J(\Sigma)=0$, complex conjugation gives $J^\dagger\Sigma=0$, and $\Sigma=P_J\Sigma=JJ^\dagger\Sigma=0$. Proposition~\ref{prop:rectangular_column} therefore turns $Q_J(G)=\Sigma$ into the single-column constraint $G|\widetilde\Omega_J\rangle=|q\rangle$.  Moreover,
\begin{equation}
\langle q|\widetilde\Omega_J\rangle = p.
\end{equation}
The projection argument in Eq.~\eqref{eq:minimal_lift_proof}, with $|\Omega\rangle$ replaced by $|\widetilde\Omega_J\rangle$, proves Eqs.~\eqref{eq:rectangular_minimal_lift} and \eqref{eq:rectangular_all_lifts}.  Summing the least lifts and using Eq.~\eqref{eq:rectangular_normalization} proves the two domination criteria. Finally, the two directions of the compression--gluing proof use only normalization, least lifts, and rank-one exposure of two deterministic processes in the same fiber.  For completeness, if $W\in\mathfrak P$ has rank $r$, then
\begin{equation}
|r_{W,J}\rangle = W^{1/2}|\widetilde\Omega_J\rangle
\end{equation}
has unit norm by Eq.~\eqref{eq:rectangular_normalization}.  Extend it to an orthonormal basis of $\operatorname{supp}W$ and apply the Fourier rotation in Lemma~1 of the main text.  This gives a rank-one decomposition in which every branch has reference weight $1/r$ and equals its own $M_J$-lift.  Hence the proof of Theorem~1 of the main text applies with $Q_+,M,R_x,|\Omega\rangle$ replaced by $Q_J,M_J,R_x^J,|\widetilde\Omega_J\rangle$.
\end{proof}

\begin{proposition}[Obstruction below the input dimension]
\label{prop:rectangular_no_go}
If $d_{X_O}<d_{X_I}$ for at least one laboratory $X$, no factorwise single-Kraus trace-preserving reference of the form used above exists. More strongly, there is no local operator $J_X:X_I\to X_O$ satisfying
\begin{equation}
J_X^\dagger J_X=cI_{X_I}
\end{equation}
for any fixed $c>0$.
\end{proposition}

\begin{proof}
The left-hand side has rank at most $d_{X_O}$, whereas the right-hand side has rank $d_{X_I}$.  The assumed dimension inequality makes equality impossible.
\end{proof}

Proposition~\ref{prop:rectangular_no_go} is a no-go for this local, rank-one-reference architecture, not for every conceivable compression. A different situation arises for one fixed multi-Kraus dilation, where several columns are retained coherently inside a chosen Stinespring realization.  Let $\mathcal H_P=A_I\otimes A_O\otimes B_I\otimes B_O$ and let $R:\mathbb C^r\to\mathcal H_P$ collect coherent reference columns, so the reference Choi operator is $C_0=RR^\dagger$, and suppose the retained data are the column array $Q=GR$.  The positive-lift feasibility conditions, least lift, and all positive lifts are exactly Eqs.~\eqref{eq:joint_lift_feasibility}--\eqref{eq:joint_all_lifts} in Theorem~\ref{thm:joint_positive_lift}, with $R$ here interpreted as the coherent column matrix.  If $C_0$ is a deterministic product reference channel and $W$ is a deterministic process, then
\begin{equation}
\operatorname{Tr}(R^\dagger WR) = \operatorname{Tr}(WC_0) =1.
\end{equation}

This coherent multi-column statement is not, by itself, a canonical QSOST generalization.  The array $GR$ presupposes access to a chosen Stinespring realization, including its Kraus-label coherence.  An ordinary black-box channel specifies $C_0$ but does not specify the coherent interference between that realization and a comparison arm.  Thus a fixed, trusted dilation defines a valid positive-lift problem, whereas a dilation-independent operational definition for arbitrary local dimensions requires additional structure beyond the present construction.

\section{Joint compression by several fixed references}
\label{app:multi_reference}

The fixed-reference result extends to a finite family when several reference isometries are tested separately.  Their compressed data are combined classically across reference settings; no coherent access to Kraus labels is assumed.  Let
\begin{equation}
\mathcal J=\{J_1,\ldots,J_K\}, \qquad J_\ell=J_{A,\ell}\otimes J_{B,\ell}, \qquad K\ge1,
\end{equation}
be a finite family of fixed factorwise reference isometries satisfying the dimension condition in Eq.~\eqref{eq:rectangular_dimension_condition}. Define
\begin{equation}
r_\ell=|\widetilde\Omega_{J_\ell}\rangle, \qquad R=[r_1\ \cdots\ r_K],
\label{eq:joint_reference_matrix}
\end{equation}
and the joint compression
\begin{equation}
\mathbf Q_{\mathcal J}(G) = \bigl(Q_{J_1}(G),\ldots,Q_{J_K}(G)\bigr).
\label{eq:joint_reference_compression}
\end{equation}
By Proposition~\ref{prop:rectangular_column}, every component $Q_{J_\ell}(G)$ is left supported, and the map $q_{J_\ell}$ is injective on that supported subspace.  Applying it to component $\ell$ gives $Gr_\ell$. Thus the joint data in Eq.~\eqref{eq:joint_reference_compression} are equivalent to the column array
\begin{equation}
\mathcal C_R(G)=GR=[Gr_1\ \cdots\ Gr_K].
\label{eq:joint_visible_columns}
\end{equation}

\begin{theorem}[Positive lift of a joint column array]
\label{thm:joint_positive_lift}
Let $Q:\mathbb C^K\to\mathcal H_P$ be a proposed visible column array and set
\begin{equation}
P=R^\dagger Q.
\label{eq:joint_gram_data}
\end{equation}
There exists $G\succeq0$ satisfying $GR=Q$ if and only if
\begin{equation}
P=P^\dagger\succeq0, \qquad \ker P\subseteq\ker Q.
\label{eq:joint_lift_feasibility}
\end{equation}
Whenever these conditions hold, the least positive lift is
\begin{equation}
M_R(Q)=QP^+Q^\dagger,
\label{eq:joint_least_lift}
\end{equation}
and every positive lift has the unique form
\begin{equation}
G=M_R(Q)+H, \qquad H\succeq0, \qquad HR=0.
\label{eq:joint_all_lifts}
\end{equation}
\end{theorem}

\begin{proof}
Suppose first that $G\succeq0$ and $GR=Q$, and put $A=G^{1/2}R$.  Then
\begin{equation}
P=R^\dagger GR=A^\dagger A\succeq0.
\end{equation}
It is Hermitian, and $\ker P=\ker A$.  Moreover $Q=G^{1/2}A$, so $Ac=0$ implies $Qc=0$, proving the kernel inclusion. Using $Q=G^{1/2}A$ in Eq.~\eqref{eq:joint_least_lift} gives
\begin{align}
G-M_R(Q) &= G^{1/2} \left[I-A(A^\dagger A)^+A^\dagger\right] G^{1/2} \succeq0,
\label{eq:joint_projection_remainder}
\end{align}
because the bracket is the orthogonal projector onto $\ker A^\dagger$.  Thus $M_R(Q)$ is below every positive lift.

Conversely, assume Eq.~\eqref{eq:joint_lift_feasibility}.  Hermiticity of $P=R^\dagger Q$ gives $Q^\dagger R=P$, and the kernel inclusion is equivalent in finite dimensions to $Q=QP^+P$.  Hence
\begin{equation}
M_R(Q)R =QP^+Q^\dagger R =QP^+P =Q.
\end{equation}
The operator $M_R(Q)$ is positive, so it is a lift.  For any positive lift $G$, Eq.~\eqref{eq:joint_projection_remainder} defines $H=G-M_R(Q)\succeq0$, and $HR=GR-M_R(Q)R=0$.  Conversely, every $H\succeq0$ with $HR=0$ can be added to $M_R(Q)$ without changing the columns.  Since the difference $G-M_R(Q)$ is fixed by $G$ and $Q$, the decomposition is unique.
\end{proof}

Consider now a finite conditioned family of joint data.  For each branch write
\begin{equation}
Q_{a|x} = \bigl[ q_{J_1}(\Sigma^{(1)}_{a|x})\ \cdots\ q_{J_K}(\Sigma^{(K)}_{a|x}) \bigr],
\end{equation}
where $P_{J_\ell}\Sigma^{(\ell)}_{a|x} =\Sigma^{(\ell)}_{a|x}$ for every $\ell,a,x$.  Assume the feasibility conditions in Eq.~\eqref{eq:joint_lift_feasibility}, and suppose that its unconditional joint compression is independent of the setting:
\begin{equation}
\sum_aQ_{a|x}=Q_T \quad\text{for every }x.
\label{eq:joint_common_unconditional}
\end{equation}
As for the single-reference data, deterministic normalization must also hold for every tested reference:
\begin{equation}
\langle r_\ell|Q_Te_\ell\rangle=1 \quad(\ell=1,\ldots,K),
\label{eq:joint_reference_normalization}
\end{equation}
where $\{e_\ell\}$ is the standard basis of $\mathbb C^K$. Define the settingwise joint domination operator
\begin{equation}
\mathcal R_x^{\mathcal J} = \sum_a M_R(Q_{a|x}).
\label{eq:joint_domination_operator}
\end{equation}

\begin{proposition}[Joint-reference domination criterion]
\label{prop:joint_domination}
Setting $x$ has positive branch lifts with a deterministic parent exactly when some $W_x\in\mathfrak P$ obeys
\begin{equation}
W_x\succeq\mathcal R_x^{\mathcal J}.
\label{eq:joint_settingwise_domination}
\end{equation}
The full family has one common deterministic parent exactly when some $W\in\mathfrak P$ obeys
\begin{equation}
W\succeq\mathcal R_x^{\mathcal J} \quad\text{for every }x.
\label{eq:joint_common_domination}
\end{equation}
\end{proposition}

\begin{proof}
Every positive lift of $Q_{a|x}$ dominates its least lift by Theorem~\ref{thm:joint_positive_lift}; summing proves necessity.  Conversely, if Eq.~\eqref{eq:joint_settingwise_domination} holds, set
\begin{equation}
H_x=W_x-\mathcal R_x^{\mathcal J}\succeq0.
\end{equation}
For each $\ell$, deterministic-process normalization, Eq.~\eqref{eq:joint_reference_normalization}, and $\mathcal R_x^{\mathcal J}R=\sum_aQ_{a|x}=Q_T$ give
\begin{equation}
\langle r_\ell|H_x|r_\ell\rangle=1-1=0.
\end{equation}
Positivity therefore implies $H_xr_\ell=0$ for every $\ell$, or $H_xR=0$. Choose any $\lambda_{a|x}\ge0$ with $\sum_a\lambda_{a|x}=1$.  The operators
\begin{equation}
G_{a|x} =M_R(Q_{a|x})+\lambda_{a|x}H_x
\end{equation}
are positive, have columns $Q_{a|x}$, and sum to $W_x$.  The identical argument with one $W$ satisfying Eq.~\eqref{eq:joint_common_domination} proves the common-parent statement.
\end{proof}

\begin{theorem}[Exact joint-reference faithfulness criterion]
\label{thm:joint_reference_faithfulness}
The following statements are equivalent:
\begin{enumerate}
\item the joint compression
$\mathbf Q_{\mathcal J}|_{\mathfrak P}$, equivalently $W\mapsto WR$ on $\mathfrak P$, is injective;
\item every finite settingwise-realizable joint-reference conditioned family
obeying Eq.~\eqref{eq:joint_common_unconditional} has one common deterministic parent;
\item the reference span separates the traceless valid-process space, in the
sense that
\begin{equation}
\left.
\begin{aligned}
&\Delta=\Delta^\dagger,\qquad L_V(\Delta)=\Delta,\\
&\operatorname{Tr}\Delta=0,\qquad \Delta R=0
\end{aligned}
\right\} \quad\Longrightarrow\quad \Delta=0.
\label{eq:joint_separating_subspace}
\end{equation}
\end{enumerate}
\end{theorem}

\begin{proof}
Suppose first that the joint compression is injective.  In a settingwise realization, Eq.~\eqref{eq:joint_common_unconditional} says $W_xR=Q_T$ for every $x$.  Injectivity gives one parent $W_x=W$ for all settings, so the same branch lifts already form a common realization.

For the converse, suppose that distinct $W_0,W_1\in\mathfrak P$ satisfy $W_0R=W_1R$.  We first construct rank-one decompositions in which every branch sees the reference span.  If $W$ has rank $r$, let
\begin{equation}
A_W=W^{1/2}R.
\end{equation}
Each factorwise $J_\ell$ is a deterministic product reference, so Eq.~\eqref{eq:rectangular_normalization} gives
\begin{equation}
\|A_W\|_{\mathrm F}^2 =\operatorname{Tr}(R^\dagger WR) =\sum_{\ell=1}^K\langle r_\ell|W|r_\ell\rangle =K>0.
\label{eq:joint_nonzero_reference_overlap}
\end{equation}
Choose any orthonormal basis of $\operatorname{supp}W$ and mix it by a unitary $U\in U(r)$.  For a fixed output basis vector $|u_a\rangle$, the bad condition $\langle u_a|A_W=0$ says that $|u_a\rangle$ lies in the proper subspace $\ker A_W^\dagger$.  It is a Haar-measure-zero condition on $U$. The union of these $r$ bad sets still has measure zero, so a unitary mixing exists for which
\begin{equation}
\langle u_a|A_W\ne0 \quad\text{for every }a.
\label{eq:joint_generic_mixing}
\end{equation}
With $|v_a\rangle=W^{1/2}|u_a\rangle$, completeness of the mixed basis gives
\begin{equation}
W=\sum_{a=1}^r|v_a\rangle\langle v_a|, \qquad c_a=R^\dagger|v_a\rangle\ne0.
\end{equation}
The visible columns of branch $a$ are $Q_a=|v_a\rangle c_a^\dagger$, and $R^\dagger Q_a=c_ac_a^\dagger$.  Since $c_a\ne0$,
\begin{equation}
M_R(Q_a) = |v_a\rangle c_a^\dagger (c_ac_a^\dagger)^+ c_a\langle v_a| =|v_a\rangle\langle v_a|.
\label{eq:joint_rank_one_exposure}
\end{equation}
Here the displayed adjacent factors denote ordinary multiplication, and $c_a^\dagger(c_ac_a^\dagger)^+c_a=1$ proves the last equality.

Apply this construction separately to $W_0$ and $W_1$, use the parent label as the setting, and pad the smaller rank-one decomposition with zero outcomes.  The family is settingwise realizable and its unconditional columns are common because $W_0R=W_1R$.  If it had a common parent $W$, the least-lift property and Eq.~\eqref{eq:joint_rank_one_exposure} would give $W\succeq W_x$ for $x=0,1$.  Both operators have trace $d_O$, so $W-W_x\succeq0$ with zero trace forces $W=W_x$, contradicting $W_0\ne W_1$.  Thus universal common-parent faithfulness implies injectivity.

It remains to prove the separating-space equivalence.  If two deterministic processes have the same columns, their difference is a Hermitian, traceless, valid-process operator annihilating $R$, so Eq.~\eqref{eq:joint_separating_subspace} implies injectivity.  Conversely, suppose a nonzero $\Delta$ satisfies the four conditions in Eq.~\eqref{eq:joint_separating_subspace}.  Let
\begin{equation}
W_*=\frac{I}{d_I}, \qquad d_I=d_{A_I}d_{B_I}.
\end{equation}
This is a full-rank deterministic process: $L_V(W_*)=W_*$ and $\operatorname{Tr}W_*=d_O$.  For sufficiently small $\epsilon>0$, $W_\pm=W_*\pm\epsilon\Delta$ remain positive, valid, and trace $d_O$, while $W_+R=W_-R$ because $\Delta R=0$.  They are distinct same-compression deterministic processes, so the joint map is not injective.
\end{proof}

\begin{proposition}[Explicit finite restoration in equal dimensions]
\label{prop:joint_unitary_restoration}
Let $d_{A_I}=d_{A_O}=m$ and $d_{B_I}=d_{B_O}=n$.  Choose orthogonal local unitary operator bases
\begin{align}
\operatorname{Tr}[(U_\mu^A)^\dagger U_{\mu'}^A] &=m\delta_{\mu\mu'},\\
\operatorname{Tr}[(V_\nu^B)^\dagger V_{\nu'}^B] &=n\delta_{\nu\nu'},
\end{align}
and take the $m^2n^2$ product references $J_{\mu\nu}=U_\mu^A\otimes V_\nu^B$.  Their reference vectors form an orthogonal basis of $\mathcal H_P$, and every Hermitian operator $G$ is reconstructed from the joint columns by
\begin{equation}
G = \frac1{mn} \sum_{\mu,\nu} G|\widetilde\Omega_{\mu\nu}\rangle \langle\widetilde\Omega_{\mu\nu}|.
\label{eq:joint_unitary_reconstruction}
\end{equation}
Consequently this finite family makes the joint compression faithful on the Hermitian operator space and, in particular, common-parent faithful.
\end{proposition}

\begin{proof}
Vectorization preserves the Hilbert--Schmidt inner product, so factorization of the product references gives
\begin{align}
&\langle\widetilde\Omega_{\mu\nu}| \widetilde\Omega_{\mu'\nu'}\rangle \nonumber\\
&\quad= \operatorname{Tr}[(U_\mu^A)^\dagger U_{\mu'}^A]\, \operatorname{Tr}[(V_\nu^B)^\dagger V_{\nu'}^B] \nonumber\\
&\quad= mn\,\delta_{\mu\mu'}\delta_{\nu\nu'}.
\label{eq:joint_reference_orthogonality}
\end{align}
There are $m^2n^2=\dim\mathcal H_P$ such nonzero orthogonal vectors, hence they form a basis and resolve the identity as
\begin{equation}
I_{\mathcal H_P} = \frac1{mn} \sum_{\mu,\nu} |\widetilde\Omega_{\mu\nu}\rangle \langle\widetilde\Omega_{\mu\nu}|.
\end{equation}
Multiplying this identity by $G$ proves Eq.~\eqref{eq:joint_unitary_reconstruction}.  Thus equality of all joint columns for Hermitian operators implies equality of the operators themselves.  This construction is sufficient; no minimal-reference count is asserted.
\end{proof}

\subsection{The minimum number of product-unitary references for qubit processes}
\label{app:minimum_five}

The full-operator reconstruction above is a sufficient construction.  For deterministic-process identification on qubit slots, fewer references suffice, and their minimum can be determined exactly.  Throughout this subsection $d_{A_I}=d_{A_O}=d_{B_I}=d_{B_O}=2$, references are fixed local product unitaries, and the joint data are exact.  Define the real vector space
\begin{equation}
\mathcal V_0=\{\Delta=\Delta^\dagger:L_V(\Delta)=\Delta, \ \operatorname{Tr}\Delta=0\}.
\label{eq:qubit_vzero}
\end{equation}

\begin{theorem}[Exactly five references for deterministic qubit processes]
\label{thm:minimum_five_qubit}
For every four references $J_\ell=U_\ell\otimes V_\ell$, with $U_\ell,V_\ell\in U(2)$, the space
\begin{equation}
\{\Delta\in\mathcal V_0:\Delta r_\ell=0\quad(\ell=1,\ldots,4)\}
\label{eq:four_reference_kernel}
\end{equation}
has real dimension at least six.  Consequently their joint compression is not injective on deterministic processes.  In contrast, the five references
\begin{equation}
\mathcal J_5=\{I\otimes I,\ Z\otimes Y,\ Y\otimes Z, \ Y\otimes X,\ X\otimes I\}
\label{eq:five_reference_pairs_supp}
\end{equation}
give an injective map $W\mapsto WR$ on deterministic processes.  Hence the minimum number of fixed local product-unitary reference groups for deterministic-process identification from exact joint data is five.
\end{theorem}

\begin{proof}[Universal four-reference lower bound]
For a local matrix let $\operatorname{vec}(U)_{ab}=U_{ba}$, consistently with the reference convention in Eq.~\eqref{eq:rectangular_reference_vector}. Let $F$ be the unitary with columns
\begin{equation}
F=\frac1{\sqrt2} \bigl[\operatorname{vec}(I)\ \operatorname{vec}(iX)\ \operatorname{vec}(iY)\ \operatorname{vec}(iZ)\bigr], \qquad T=F\otimes F.
\label{eq:reference_magic_basis}
\end{equation}
Every $U\in U(2)$ can be written
\begin{equation}
U=e^{i\alpha}\left(a_0 I+i\sum_{j=1}^3a_j\sigma_j\right), \qquad a_j\in\mathbb R,\quad\sum_{j=0}^3a_j^2=1.
\label{eq:unitary_quaternion_reference}
\end{equation}
Indeed, divide by a square root of $\det U$ and use the standard form $\left(\begin{smallmatrix}z&w\\-w^*&z^*\end{smallmatrix}\right)$ of an $SU(2)$ matrix.  Writing $V$ with real coefficients $b_j$ in the same way gives, in the interleaved order,
\begin{equation}
T^\dagger r_{U\otimes V} =2e^{i(\alpha+\beta)}(a\otimes b).
\label{eq:all_product_references_real}
\end{equation}
Removing each nonzero scalar phase leaves four real vectors.  Their real span $S\subseteq\mathbb R^{16}$ has dimension $r\le4$, including when the references are dependent.  This step covers every local product unitary, not merely Pauli references.

The valid Pauli support patterns in Eq.~\eqref{eq:general_valid_patterns} with odd weight are $0010,1000,1011,1110$.  Their real span $\mathcal V_{\mathrm{odd}}\subset\mathcal V_0$ has dimension $3+3+27+27=60$.  In the basis $T$, each such Hermitian matrix is purely imaginary and hence equals $iA$ with $A$ real skew-symmetric.  To see this for the entire span, let $K$ denote computational-basis conjugation. The antiunitary $-(Y\otimes Y)K$ fixes every column of $F$.  Since $Y\sigma_j^*Y=-\sigma_j$ for $j=1,2,3$, its conjugation multiplies a two-factor Pauli string by its weight parity.  On both laboratories, $TKT^\dagger$ therefore multiplies a four-factor Pauli string of weight $w$ by $(-1)^w$.  Thus
\begin{equation}
\overline{T^\dagger B T}=(-1)^wT^\dagger B T, \qquad -iT^\dagger\mathcal V_{\mathrm{odd}}T \subseteq\mathfrak{so}(16).
\label{eq:odd_valid_skew_sector}
\end{equation}
The latter is a real $60$-dimensional subspace of the $120$-dimensional space of real skew-symmetric matrices.

For any real $r$-dimensional $S$, the restriction map $A\mapsto A|_S$ on $\mathfrak{so}(16)$ has rank
\begin{equation}
120-\frac{(16-r)(15-r)}2 =16r-\frac{r(r+1)}2.
\label{eq:skew_restriction_rank}
\end{equation}
Choose an orthonormal basis beginning with $S$: a skew matrix vanishing on $S$ has zero first $r$ columns and rows, and its remaining block is an arbitrary skew matrix.  This proves Eq.~\eqref{eq:skew_restriction_rank} for every $S$, without a genericity assumption.  The restriction to the $60$-dimensional valid odd subspace consequently has kernel dimension at least
\begin{equation}
60-\left(16r-\frac{r(r+1)}2\right)\ge6\qquad(r\le4).
\label{eq:four_reference_six_hidden}
\end{equation}
Pulling any nonzero element back by $T$ gives a nonzero $\Delta\in\mathcal V_0$ with $\Delta r_\ell=0$ for all four references. In particular,
\begin{equation}
W_\pm=\frac{I_{16}}4\pm\epsilon\Delta, \qquad 0<\epsilon<\frac1{4\|\Delta\|_\infty},
\label{eq:four_reference_distinct_processes}
\end{equation}
are distinct positive-definite valid processes of trace four with identical joint data.  Repetition covers fewer than four references as well.
\end{proof}

\begin{proof}[Five-reference upper bound and exact certificate]
Let $\{B_j\}_{j=0}^{86}$ consist of the traceless valid Pauli strings in lexicographic order $I<X<Y<Z$ on $A_IA_OB_IB_O$.  Equivalently, take all strings with support patterns
\begin{equation}
0010,\ 0110,\ 1000,\ 1001,\ 1010,\ 1011,\ 1110.
\label{eq:five_certificate_basis}
\end{equation}
Their multiplicities $3,9,3,9,9,27,27$ sum to $87$.  Orthogonality of Pauli strings makes them a real basis of $\mathcal V_0$.

Index the five references of Eq.~\eqref{eq:five_reference_pairs_supp} by $\ell=0,\ldots,4$ in the displayed order.  Set $r_\ell[a,b,c,d]=(U_\ell)_{ba}(V_\ell)_{dc}$, with binary component index $t=8a+4b+2c+d$.  Define the integer real matrix $\mathsf A\in\mathbb Z^{160\times87}$ by
\begin{align}
\mathsf A_{32\ell+t,j}&=\operatorname{Re}(B_jr_\ell)_t, \nonumber\\
\mathsf A_{32\ell+16+t,j}&=\operatorname{Im}(B_jr_\ell)_t, \qquad t=0,\ldots,15.
\label{eq:five_certificate_matrix}
\end{align}
All row and column indices in this certificate are zero-based.  The following ordered $87$ rows, with all columns in their stated order, form a square minor $\mathsf A_{\mathcal I,:}$:
\begin{equation}
\begin{aligned}
\mathcal I=(&0,1,2,3,4,5,6,7,8,9,10,12,13,16,17,18,\\
&20,21,22,23,24,25,26,29,32,33,35,36,37,38,39,40,\\
&41,43,48,49,50,52,53,54,55,56,57,59,60,61,64,65,\\
&66,67,68,69,70,73,77,80,81,82,83,84,85,86,87,88,\\
&89,93,96,97,98,99,100,101,112,113,114,115,116,117,\\
&128,129,132,133,134,144,145,149,153).
\end{aligned}
\label{eq:five_certificate_rows}
\end{equation}
Exact integer arithmetic gives
\begin{equation}
\det\mathsf A_{\mathcal I,:}=-2^{121}\ne0.
\label{eq:five_certificate_determinant}
\end{equation}
Equations~\eqref{eq:five_certificate_basis}--\eqref{eq:five_certificate_rows} specify the minor completely and allow its determinant to be checked without any fitted data or numerical rank tolerance.  The accompanying script \texttt{five\_reference\_certificate.py} reconstructs the matrix and verifies the saved minor using exact Bareiss elimination; a separate implementation, \texttt{review\_reference\_bounds.py}, constructs the process projector by computational-basis partial traces and reproduces the same matrix and determinant.

For $\Delta=\sum_jc_jB_j\in\mathcal V_0$, the conditions $\Delta r_\ell=0$ are exactly $\mathsf A c=0$ over the reals.  The nonzero maximal minor forces $c=0$.  Theorem~\ref{thm:joint_reference_faithfulness} then proves injectivity on deterministic processes and universal common-parent faithfulness for finite settingwise-realizable joint families. Together with the lower bound, this proves Theorem~\ref{thm:minimum_five_qubit}.
\end{proof}

This minimum concerns deterministic processes on the declared qubit slots, not unrestricted positive branch operators.  Indeed the five reference vectors span only a five-dimensional subspace of $\mathbb C^{16}$. For any nonzero $z\in\ker R^\dagger$, the positive operator $H=|z\rangle\langle z|\ne0$ obeys $HR=0$, so the unrestricted branch map retains invisible positive freedom.  Nor are the five groups five measurement settings: each reference group supplies its complete joint compression data.  No minimal count in other dimensions, approximate-data guarantee, or causal-order certification follows from this theorem.

\section*{Supplemental References}

\begin{enumerate}

\item[{[S1]}]
M.-D. Choi, Completely positive linear maps on complex matrices, Linear Algebra Appl. \textbf{10}, 285 (1975).

\item[{[S2]}]
J. Bavaresco, M. Ara\'ujo, C. Brukner, and M. T. Quintino, Semi-device-independent certification of indefinite causal order, Quantum \textbf{3}, 176 (2019).

\item[{[S3]}]
L. P. Hughston, R. Jozsa, and W. K. Wootters, A complete classification of quantum ensembles having a given density matrix, Phys. Lett. A \textbf{183}, 14 (1993).

\item[{[S4]}]
M. Baroni, D. Leichtle, I. \v{S}upi\'c, D. Markham, and M. T. Quintino, Composable simultaneous purification: When all communication scenarios reduce to spatial correlations, arXiv:2601.05158 [quant-ph] (2026).

\item[{[S5]}]
C.-Y. Hsieh, M. Lostaglio, and A. Ac\'in, Quantum channel marginal problem, Phys. Rev. Research \textbf{4}, 013249 (2022).

\item[{[S6]}]
O. Oreshkov, F. Costa, and C. Brukner, Quantum correlations with no causal order, Nat. Commun. \textbf{3}, 1092 (2012).

\item[{[S7]}]
G. Chiribella, G. M. D'Ariano, and P. Perinotti, Theoretical framework for quantum networks, Phys. Rev. A \textbf{80}, 022339 (2009).

\item[{[S8]}]
Z. Jia, M. Song, and D. Kaszlikowski, Quantum space-time marginal problem: Global causal structure from local causal information, New J. Phys. \textbf{25}, 123038 (2023).

\item[{[S9]}]
W. F. Stinespring, Positive functions on $C^*$-algebras, Proc. Am. Math. Soc. \textbf{6}, 211 (1955).

\item[{[S10]}]
M. Ara\'ujo, A. Feix, M. Navascu\'es, and C. Brukner, A purification postulate for quantum mechanics with indefinite causal order, Quantum \textbf{1}, 10 (2017).

\item[{[S11]}]
S. H. Lie and H. Kwon, Probing Quantum States over Spacetime through Interferometry, Phys. Rev. Lett. \textbf{136}, 250201 (2026).

\end{enumerate}

\end{document}